\newif\ifanonymous
\anonymousfalse

\documentclass[runningheads]{llncs}
\usepackage{graphicx} % Required for inserting images
\usepackage{xspace}

\usepackage{float,siunitx} 
\usepackage{amsmath,amssymb,amsfonts}
\usepackage{booktabs,caption,subcaption}
\usepackage{amsmath}
\usepackage{algorithm}
\usepackage{algorithmic}
\usepackage{graphics,graphicx,xcolor}
\usepackage[utf8]{inputenc}
\usepackage{authblk}
\usepackage{diagbox}
\definecolor{honestblue}{RGB}{42,102,204}
\definecolor{adversaryred}{RGB}{204,32,32}
\definecolor{altblue}{RGB}{30,140,220}
\definecolor{bribeegreen}{RGB}{50,160,60}
\definecolor{evilorange}{RGB}{240,80,30}
\definecolor{evilgreen}{RGB}{51,184,64}

\newcommand{\advblock}[2]{\textcolor{red}{\ensuremath{\mathbf{A}}_{#1}^{#2}}}
\newcommand{\honestblock}[2]{\textcolor{blue}{\ensuremath{\mathbf{H}}_{#1}^{#2}}}

\newcommand{\offerbribe}{{\ensuremath{\mathsf{offerBribe}}}}
\newcommand{\takebribe}{{\ensuremath{\mathsf{takeBribe}}}}
\newcommand{\cf}{\emph{cf.}\xspace}

\newcommand{\eg}{\textit{e.g.,} }
\newcommand{\ie}{\textit{i.e.,} }

\usepackage{listings}
\usepackage{listings-solidity} % copy listings-solidity.sty to your project

\usepackage{pdflscape}

\usepackage{xcolor}
\definecolor{verylightgray}{rgb}{.97,.97,.97}

\usepackage{listings, xcolor}

\definecolor{verylightgray}{rgb}{.97,.97,.97}

\lstdefinelanguage{Solidity}{
	keywords=[1]{anonymous, assembly, assert, balance, break, call, callcode, case, catch, class, constant, continue, constructor, contract, debugger, default, delegatecall, delete, do, else, emit, event, experimental, export, external, false, finally, for, function, gas, if, implements, import, in, indexed, instanceof, interface, internal, is, length, library, log0, log1, log2, log3, log4, memory, modifier, new, payable, pragma, private, protected, public, pure, push, require, return, returns, revert, selfdestruct, send, solidity, storage, struct, suicide, super, switch, then, this, throw, transfer, true, try, typeof, using, value, view, while, with, addmod, ecrecover, keccak256, mulmod, ripemd160, sha256, sha3}, % generic keywords including crypto operations
	keywordstyle=[1]\color{blue}\bfseries,
	keywords=[2]{address, bool, byte, bytes, bytes1, bytes2, bytes3, bytes4, bytes5, bytes6, bytes7, bytes8, bytes9, bytes10, bytes11, bytes12, bytes13, bytes14, bytes15, bytes16, bytes17, bytes18, bytes19, bytes20, bytes21, bytes22, bytes23, bytes24, bytes25, bytes26, bytes27, bytes28, bytes29, bytes30, bytes31, bytes32, enum, int, int8, int16, int24, int32, int40, int48, int56, int64, int72, int80, int88, int96, int104, int112, int120, int128, int136, int144, int152, int160, int168, int176, int184, int192, int200, int208, int216, int224, int232, int240, int248, int256, mapping, string, uint, uint8, uint16, uint24, uint32, uint40, uint48, uint56, uint64, uint72, uint80, uint88, uint96, uint104, uint112, uint120, uint128, uint136, uint144, uint152, uint160, uint168, uint176, uint184, uint192, uint200, uint208, uint216, uint224, uint232, uint240, uint248, uint256, var, void, ether, finney, szabo, wei, days, hours, minutes, seconds, weeks, years},	% types; money and time units
	keywordstyle=[2]\color{teal}\bfseries,
	keywords=[3]{block, blockhash, coinbase, difficulty, gaslimit, number, timestamp, msg, data, gas, sender, sig, value, now, tx, gasprice, origin},	% environment variables
	keywordstyle=[3]\color{violet}\bfseries,
	identifierstyle=\color{black},
	sensitive=true,
	comment=[l]{//},
	morecomment=[s]{/*}{*/},
	commentstyle=\color{gray}\ttfamily,
	stringstyle=\color{red}\ttfamily,
	morestring=[b]',
	morestring=[b]"
}

\usepackage{url}
\usepackage{pst-node}% http://ctan.org/pkg/pst-node
\usepackage{multido}% http://ctan.org/pkg/multido

\usepackage{varioref,hyperref}
\hypersetup{
  linkcolor  = red,
  citecolor  = green,
  urlcolor   = blue,
  colorlinks = true,
}
\usepackage[nameinlink,capitalize]{cleveref}

\usepackage{tikz}
\usetikzlibrary{positioning,chains,fit,shapes,calc}
\usepackage[most]{tcolorbox}
\usepackage{comment}
\usepackage[section]{placeins}

\usepackage[utf8]{inputenc}
\usepackage{pgfplots}
\DeclareUnicodeCharacter{2212}{−}
\usepgfplotslibrary{groupplots,dateplot}
\usetikzlibrary{patterns,shapes.arrows}
\pgfplotsset{compat=newest}

\usepackage{tikz}
\usetikzlibrary{arrows.meta,positioning,calc,fit}

\usepackage{pifont}
\newcommand{\cmark}{\checkmark} % checkmark
\newcommand{\xmark}{\ding{55}}  % x mark

\tikzset{
  block/.style = {draw, thick, rounded corners=2pt, minimum width=1.7cm, minimum height=.8cm, align=center},
  honest/.style = {block, fill=honestblue!85, text=white},
  adversary/.style = {block, draw=adversaryred, dashed, very thick, fill=adversaryred!15, text=adversaryred},
  honestarw/.style = {altblue, -{Latex[length=3mm]}, very thick},
  advarw/.style = {adversaryred, dashed, -{Latex[length=3mm]}, thick},
  altarw/.style = {bribeegreen, dash dot, -{Latex[length=3mm]}, thick},
achain/.style = {black, dashed, -{Latex[length=3mm]}, thick},
chain/.style = {black, -{Latex[length=3mm]}, thick},
slotline/.style = {gray!55, line width=0.6pt},
}

\usepackage{makecell}

\usepackage{mathtools}

\newcommand{\pk}{\mathsf{pk}}
\newcommand{\sk}{\mathsf{sk}}

\newcommand{\pboost}{{\ensuremath{\mathit{p}_{\mathsf{boost}}}}}

\newcommand{\glabel}[1]{\{X^b\}}

\newcommand{\bls}{\mathsf{BLS}}
\newcommand{\signn}{\mathsf{Sign}}
\newcommand{\verifyy}{\mathsf{Verify}}
\newcommand{\keygenn}{\mathsf{KeyGen}}
\def\secparam{\ensuremath{1^\secpar}}
\def\secpar{\ensuremath{\lambda}}

\usepackage{xcolor}

\newif\ifrev
\revtrue   % or \revfalse

\newcommand{\revision}[1]{%
  \ifrev
    \textcolor{black}{#1}%
  \else
    #1%
  \fi
}
\title{\emph{Bribers, Bribers on The Chain,\\
Is Resisting All in Vain?}\\ \large Trustless Consensus Manipulation \\ Through Bribing Contracts}
\ifanonymous
    \author{Bribers, Bribers\inst{}}
    \institute{On The Chain} 
\else
\author{Bence Soóki-Tóth\inst{1,2}\footnote{Bence Soóki-Tóth and István András Seres contributed equally to this work.}, István András Seres$^{\star\star}$\inst{1}, Kamilla Kara\inst{1}, Ábel Nagy\inst{1}, Balázs Pejó\inst{3,4}, Gergely Biczók\inst{3,4}}
\institute{Eötvös Loránd University\and Aarhus University\and CrySyS Lab, Dep. of Networked Systems and Services, Budapest University of Technology and Economics \and HUN-REN-BME Information Systems Research Group}
\authorrunning{Soóki-Tóth et al.}
\fi
\titlerunning{Bribers, Bribers on The Chain, Is Resisting All in Vain?}

\hypersetup{
    colorlinks=true,
    linkcolor=blue,
    filecolor=magenta,      
    urlcolor=blue,
    citecolor=blue,
    pdftitle={Bribers, Bribers on The Chain, Is Resisting All in Vain? Trustless Consensus Manipulation Through Bribing Contracts} 
}

\begin{document}

\maketitle

\begin{abstract}
The long-term success of cryptocurrencies largely depends on the incentive compatibility provided to the validators. Bribery attacks, facilitated trustlessly via smart contracts, threaten this foundation. This work introduces, implements, and evaluates three novel and efficient bribery contracts targeting Ethereum validators. The first bribery contract enables a briber to fork the blockchain by buying votes on their proposed blocks. The second contract incentivizes validators to voluntarily exit the consensus protocol, thus increasing the adversary's relative staking power. The third contract builds a trustless bribery market that enables the briber to auction off their manipulative power over the RANDAO, Ethereum's distributed randomness beacon. Finally, we provide an initial game-theoretical analysis of one of the described bribery markets.

\end{abstract}

\section{Introduction}

The incentive compatibility of protocol participants is a central theme in cryptocurrency security research, as first analyzed by Nakamoto in the celebrated Bitcoin white paper~\cite[Section 6.]{nakamoto2008bitcoin}. The seminal work of Eyal and Sirer demonstrated the first incentive incompatibilities in the Bitcoin protocol, showing that an economically rational miner would not follow the protocol but rather employ a selfish mining strategy to increase its block rewards~\cite{eyal2018majority}. Subsequently, an extensive research direction has explored the incentives of protocol participants in various settings, such as Bitcoin without block rewards~\cite{carlsten2016instability}. This paper studies the incentives of a cryptocurrency consensus protocol where trustless, cheap, and efficient Turing-complete bribery contracts are available. In a bribery attack, a motivated briber offers a payment to \emph{economically rational validators} to deviate from the protocol, to \eg double spend~\cite{judmayer2021sok,teutsch2016cryptocurrencies}, fork the blockchain~\cite{DBLP:conf/aft/AvarikiotiKLM24,judmayer2022much}, or censor transactions~\cite{DBLP:conf/sigecom/BergerFMS25,mccorry2018smart,nadahalli2021timelocked}. Such protocol deviations can often be verified efficiently using cryptographic primitives (\eg signatures or zero-knowledge proofs), facilitating the atomic exchange of the bribe for the requested deviant behavior.

Bribery attacks have a decade-long history, initiated by Bonneau at FC'16~\cite{bonneau2016buy}. Much of this prior work has focused on bribing validators cross-chain, which considerably limits their practicality~\cite{judmayer2022much,judmayer2019pay,judmayer2021sok}. These works often assume a funding blockchain that hosts the bribery contracts and another (typically Proof-of-Work (PoW)) blockchain whose validators are being bribed, introducing limitations related to synchrony, liquidity, and interoperability. To our knowledge, bribery attacks leveraging Turing-complete scripting capabilities native to a Proof-of-Stake (PoS) consensus algorithm have not been rigorously analyzed. This setting is particularly risky because verifying properties of a PoS consensus algorithm (\eg the current head of the chain~\cite{kiayias2020non} or the active validator set) is often more efficient than in a PoW context, especially with access to expressive smart contract functionality. Note that in our closest related work, the commitment attacks devised in~\cite{DBLP:conf/eurosp/SarencheTMSP25} are not trustless, as the bribee must trust the briber.

This research gap is critical due to a common misconception about crypto-economic security. It is widely but wrongly believed that permissionless consensus protocols offer a security budget equal to the total value of staked assets ($\approx 115.3$ billion USD in Ethereum). This is only true for long-term takeovers. Contrary to this belief, in more realistic models with economically rational validators, the true security budget is significantly lower. We emphasize this by demonstrating the surprisingly low cost of the attacks we introduce. In particular, we show that an adversary can fork the Ethereum blockchain by offering bribes of less than 0.09 Ether ($\approx 334$ USD) to rational validators,~\cf~\Cref{sec:pay_to_fork_incentives}.

The attacks we present undermine the fundamental security guarantees of a cryptocurrency consensus protocol. Although cryptocurrencies pursue many security goals, their integrity rests on three core properties: safety, liveness, and fairness. Safety ensures that every honest validator agrees on the same ledger state. Liveness guarantees that all valid transactions are eventually included in the ledger. Finally, fairness dictates that a validator with $\alpha$ percent of the voting power should, in expectation, receive $\alpha$ percent of the protocol rewards.

\paragraph{Our Contributions.}
This work makes the following contributions.

\begin{description}
    \item[\textbf{Novel Bribery Contracts for PoS Ethereum.}] We design, implement, and evaluate three efficient, smart contract-based bribery attacks, each targeting a fundamental consensus property:
    \begin{itemize}
        \item \textbf{$\mathsf{PayToAttest}$:} Violates \textbf{safety} by enabling a briber to fork the blockchain by buying votes (often validator attestations) for a competing chain.
        \item \textbf{$\mathsf{PayToExit}$:} Threatens \textbf{liveness} by incentivizing validators to exit the active validator set, thereby increasing the adversary's relative stake.
        \item \textbf{$\mathsf{PayToBias}$:} Impacts \textbf{fairness} by creating a market for biasing the RANDAO, Ethereum's randomness beacon. To our knowledge, this is the first public implementation of such a bribery market.
    \end{itemize}
    We have released our implementation as an open-source artifact to foster further research.~\footnote{
    \ifanonymous
    ~\url{https://anonymous.4open.science/r/bribery-zoo-70FD/}.
    \else
    ~\url{https://github.com/0xSooki/bribery-zoo}.
    \fi}

    \item[\textbf{Game-Theoretic Analysis of Practicality.}] We model the $\mathsf{PayToExit}$ attack as a Stackelberg game to characterize the adversary's optimal bribe. Our analysis demonstrates the attack's practicality, deriving an optimal per-validator bribe of $9.23$~ETH ($\approx 34,225$ USD), an accessible amount for a motivated attacker.
\end{description}

\paragraph{Outline.} The remainder of this paper is organized as follows. 
We review the pertinent background in~\Cref{sec:preliminaries}. 
We describe our system and threat models in~\Cref{sec:system_model}. 
We introduce our three bribery contracts in~\Cref{sec:bribing_contracts} and analyze the resulting incentives in~\Cref{sec:incentives}. 
\revision{We discuss our related work in~\Cref{sec:related_work}.}
We conclude our paper with potential extensions and open questions in~\Cref{sec:conclusion}.
\section{Preliminaries}\label{sec:preliminaries}

\paragraph{Notations. }
We model the consensus environment at a given time $t$ as having $N(t)$ total active validators who collectively stake $S(t)$ amount of ETH. An adversary, or briber, controls $\alpha$ fraction of this total stake ($0 \le \alpha \le 0.5$). The remaining $(1-\alpha)$ stake belongs to non-adversarial validators, which we partition based on their behavior: a fraction $\beta$ ($0 \le \beta \le 1$) are \textit{economically rational} and will accept a profitable bribe, while the remaining fraction $(1-\alpha)(1-\beta)$ are \textit{honest} and will always follow the protocol. The virtual weight for timely block proposals is denoted by $\pboost{}$, which is set to $\pboost{}=0.4$, at the time of writing.

%We summarize our used notations in~\Cref{tab:notation}.
%\begin{table}[h!]
%\centering
%\begin{tabular}{c|l}
%\textbf{Symbol} & \textbf{Description}  \\
%\hline
%$N(t)$     & Number of active validators at time $t$  \\
%$\rewardPerEpoch{}$  & Honest reward per validator  (ETH / time unit (\eg per epoch)) \\
%$k$     & Reward constant in $\rewardPerEpoch$  estimated empirically, e.g. $2940.21$ \\
%$M(N)$     & Bribe offered by the attacker (ETH) \\
%$\pboost{}$     & Virtual weight for timely block proposals ($\pboost{}=0.4$)  \\
%$\alpha$ & The adversary's staking power ($0\leq\alpha\leq1$) \\
%$\beta$   & Economically rational validators' ratio in honest validators ($0\leq\beta\leq1$) \\
%\end{tabular}
%\caption{Notation and variable definitions used in bribery and slashing models.}
%\label{tab:notation}
%\end{table}

\subsection{The BLS signature scheme}\label{sec:background_bls_sigs}

The BLS signature scheme, proposed by Boneh, Lynn, and Shacham~\cite{boneh2001short}, is used in the Ethereum consensus protocol for validator attestations. The scheme utilizes a non-degenerate, efficiently computable, bilinear pairing function $e:\mathbb{G}_1\times\mathbb{G}_2\rightarrow\mathbb{G}_T$ over prime-order cyclic groups such that $\vert\mathbb{G}_1\vert=\vert\mathbb{G}_2\vert=\vert\mathbb{G}_T\vert=p$. It consists of the following algorithms:

\begin{description} 
     \item[$\bls.\keygenn(\secparam)$.] Samples uniformly at random a secret signing key $\sk \stackrel{\$}\leftarrow \mathbb{F}_p$ and sets the public verification key as $\pk = g^{\sk} \in \mathbb{G}_1$.
    
     \item[$\bls.\signn(\sk, m)$.] Computes the signature $\sigma := H(m)^{\sk} \in \mathbb{G}_2$, where $H: \{0, 1\}^* \rightarrow \mathbb{G}_2$ is a hash-to-curve function, and modeled as a random oracle. 
    
     \item[$\bls.\verifyy(\pk, m, \sigma)$.] Outputs $1$ if $e(\sigma, g_2) = e(\pk,H(m))$ holds, and $0$ otherwise. 
\end{description}

The BLS signature scheme is unique and provides strong unforgeability under the co-CDH assumption~\cite{boneh2001short}. The crucial property we exploit in this work is the scheme's support for signature aggregation~\cite{boneh2018bls}.
%\footnote{In Ethereum, rogue key attacks are not an issue as the possession of the secret keys $\sk$ for a validator's $\pk$ are ensured at both the consensus and application levels.} 
More precisely, if the \emph{same message} $m$ is signed under multiple public keys $\{\pk_i\}^{n}_{i=1}$, the corresponding individual signatures $\{\sigma_i\}^{n}_{i=1}$ can be combined. The resulting aggregate signature $\sigma=\sigma_1\cdot\dots\sigma_n$ and aggregate public key $\pk=\pk_1\cdot\dots\pk_n$ can be validated with a single pairing check, reducing the verification complexity to $\mathcal{O}(1)$:

\begin{equation}\label{eq:bls_same_message_check}
    e\Big(\prod^n_{i=1}\sigma_i,g_2\Big)\stackrel{?}{=}e\Big(\prod^n_{i=1}\pk_i,H(m)\Big)\enspace.
\end{equation}

The $\mathcal{O}(1)$ verification efficiency does not hold if validators sign \emph{different messages} $\{m_i\}^{n}_{i=1}$. In this scenario, the verification equation requires $n+1$ pairings and $n$ target group operations. Specifically:

\begin{equation}\label{eq:bls_different_message_check}
    e\Big(\prod^n_{i=1}\sigma_i,g_2\Big)\stackrel{?}{=}\prod^{n}_{i=1}e\Big(\pk_i,H(m_i)\Big)\enspace.
\end{equation}

\Cref{fig:bls_batch_verification_measurements} compares the on-chain gas cost of BLS batch verification for $n$ signatures on a single message $m$ versus on distinct messages $\{m_i\}_{i\in[n]}$.

\subsection{Ethereum Proof-of-Stake essentials }\label{sec:pos_ethereum_prelims}

\begin{figure}[t!]
    \centering
    \begin{minipage}{0.45\textwidth}
        \centering
        \begin{lstlisting}
    class AttestationData(Container):
    slot: Slot
    (*@\color{red}index: CommitteeIndex@*)
    # LMD GHOST vote
    beacon_block_root: Root
    # FFG vote
    source: Checkpoint
    target: Checkpoint
        \end{lstlisting}
        \caption*{\texttt{AttestationData} class}
    \end{minipage}
    \hfill
    \begin{minipage}{0.45\textwidth}
        \centering
        \begin{lstlisting}
    class VoluntaryExit(Container):
    # Epoch when exit can be processed
    epoch: Epoch 
    validator_index: ValidatorIndex
        \end{lstlisting}
        \caption*{\texttt{VoluntaryExit} class}
    \end{minipage}
    \caption{Structures of the \texttt{AttestationData} and \texttt{VoluntaryExit} classes whose instances must be signed by validators whenever they attest to a proposed block or they want to cease to be a validator. Note that the \texttt{index} field was removed from the \texttt{AttestationData} for efficiency reasons as part of the Pectra hard fork.}
    \label{fig:class_structures}
\end{figure}

The Ethereum Proof-of-Stake protocol~\cite{pavloff2023ethereum} relies on validators to secure the network by proposing and attesting to blocks. Any user can become a validator by depositing $32$~ETH into the staking contract, and can similarly exit by withdrawing their stake. Validators are compensated with protocol rewards for correctly performing their duties, but their stake can be slashed for malicious behavior such as proposing or attesting to conflicting blocks. Crucially, the bribery attacks we introduce in this work \emph{do not entail} any slashable offenses.

In Ethereum PoS, time is structured into epochs, with each epoch comprising $32$ slots of $12$ seconds each. Within the first four seconds of a slot, an assigned validator may propose a new block. During the remaining eight seconds, a committee of $\frac{N(t)}{32}$ validators is selected to attest to what they perceive as the head of the chain by BLS-signing an attestation message,~\cf~\Cref{fig:class_structures}.

If a fork occurs, the fork-choice rule dictates that the branch in the block tree with the greatest weight of attestations is considered the canonical chain. Attestation processing and rewards are described in further detail in~\Cref{sec:attestation_rewards}. Importantly for our work, validator duties (\ie block proposal and attestation roles) are publicly known at least $6.4$ minutes in advance. For a comprehensive description of the Ethereum PoS protocol, we refer the reader to~\cite{pavloff2023ethereum}.

\paragraph{Recent Protocol Upgrades}
Several recent Ethereum Improvement Proposals (EIPs) are essential for implementing the efficient bribery contracts we present. EIP-4788~\cite{stokes2022eipbeaconroot}, included in the Dencun upgrade (March 2024), exposes the beacon block root to the Ethereum Virtual Machine (EVM), greatly simplifying the on-chain verification of consensus-related statements for smart contracts.

Additionally, three EIPs from the Pectra upgrade (May 2025) provide crucial functionalities. EIP-2537~\cite{vlasov2020eip} introduces gas-efficient precompiles for pairing operations on the BLS12-381 curve. EIP-7002~\cite{ryan2023eip} allows validator exits to be initiated from the execution layer, which enables a smart contract to verify more efficiently that a validator has voluntarily exited the active set. %\bence{Ez az EIP csak azt teszi lehetővé, hogy EL-ről lehessen csinálni exitet, a miénk CL-en is működik}.
Finally, EIP-7549~\cite{dapplion2023eip} moves the index field out of the \textsf{AttestationData} structure,~\cf~\Cref{fig:class_structures}. This change ensures that when validators attest to the same block, they all $\mathsf{BLS}$-sign the exact \emph{same message}. As a result, thousands of their signatures can be aggregated off-chain and verified on-chain with high efficiency,~\cf~\Cref{eq:bls_same_message_check}.
\section{System Model}\label{sec:system_model}

Our system model, depicted in \Cref{fig:system_model}, consists of four main entities. The consensus participants within this model are categorized according to the BAR (Byzantine, Altruistic, Rational) fault model~\cite{aiyer2005bar}.

\begin{figure}[t!]
    \centering
    \includegraphics[width=12cm]{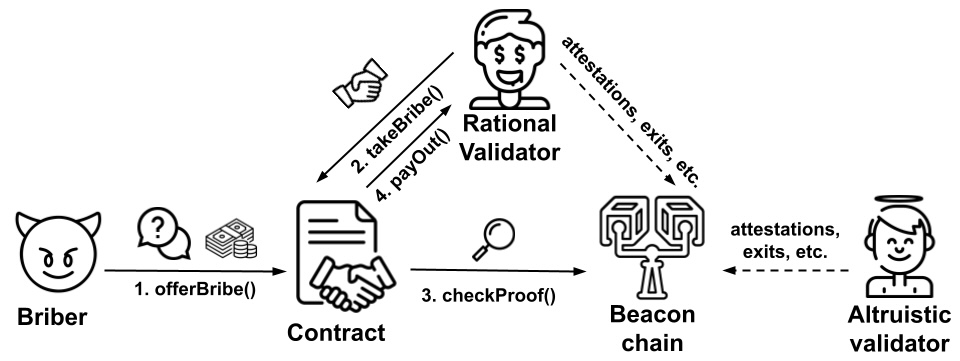}
    \caption{Our simple bribery system model. Typically, Steps $2$ and $3$ are executed atomically in the same transaction. If not, then the transaction in Step $3$ could be initiated by anyone; thus, it originates from the contract in the figure.}
    \label{fig:system_model}
\end{figure}

\begin{description}    
    \item[Briber] The briber is a potentially extrinsically motivated entity controlling an $\alpha$ fraction of the total stake. Their objective is to manipulate consensus by offering rewards to other validators for deviating from the protocol.
    \item[Rational] Rational validators, controlling a $(1-\alpha)\beta$ fraction of the stake, are profit-maximizing actors. We assume they will deviate from the protocol if an offered bribe is greater than the potential rewards from honest participation.
    \item[Altruistic] Altruistic validators, who control a $(1-\alpha)(1-\beta)$ fraction of the stake, are honest participants. They adhere strictly to the protocol at all times, regardless of any external incentives or bribes.
    \item[Contract] To facilitate an atomic, fair exchange, which is impossible without a trusted third party~\cite{pagnia1999impossibility}, we model a smart contract that acts as this trusted intermediary. This on-chain entity has a transparent state, holds the briber's funds in escrow, and automatically releases payment to any validator who provably fulfills the briber's conditions (\eg provides a specific (BLS) signature, Merkle authentication path, or other zero-knowledge proofs).
\end{description}

For our analysis, we assume the staking power of each entity remains constant, with the notable exception of the $\mathsf{PayToExit}$ model, which analyzes dynamic changes in stake distribution. We also assume a synchronous communication model where the maximum network delay is bounded at four seconds, corresponding to one-third of the duration of an Ethereum slot.

\subsection{Bribing contract interface}\label{sec:bribing_interface}

We foresee a future where consensus manipulation could become normalized and institutionalized, much like Maximal Extractable Value (MEV) is today~\cite{daian2020flash}. Such a development would necessitate standard, open interfaces to allow bribers and bribees to interact seamlessly. Thus, we propose a baseline interface, $\mathsf{IBribe}$.

We expect that rational agents will constantly monitor contracts implementing this interface by calling the view-only function $\mathsf{bribeAmnt()}$, to automatically accept any bribe exceeding their profitability threshold. This behavior mirrors the automated strategies already seen in on-chain MEV extraction~\cite{DBLP:conf/aft/SolmazHVW25}.

The $\mathsf{IBribe}$ interface,~\cf~\Cref{fig:bribe_interface}, defines two categories of functions. The core bribery interaction is handled by $\mathsf{offerBribe}()$ and $\mathsf{takeBribe}()$, which facilitate the bribe itself (\cf~\Cref{fig:system_model}). Additionally, it includes three convenience functions for fund management: $\mathsf{depositFunds}()$, $\mathsf{updateBribeAmnt}()$, and $\mathsf{withdrawFunds}()$.

\begin{figure}
\begin{minipage}{.95\linewidth}
\begin{lstlisting}[language=Solidity][float]
interface IBribe {
  function bribeAmnt() public view returns (uint256);
  function depositFunds() external payable;
  function updateBribeAmnt(uint256 amnt) external;
  function withdrawFunds(uint256 amnt) external;
  function offerBribe(bytes calldata pubkey, uint256 amnt, 
    bytes memory data) external;
  function takeBribe(bytes calldata signature) external; 
}
\end{lstlisting}
\end{minipage}
\caption{The $\mathsf{IBribe}$ interface for our proposed bribery contracts.}\label{fig:bribe_interface}
\end{figure}

\begin{description}
    \item[Offer bribe] The briber initiates the process by calling the $\mathsf{offerBribe}()$ function. In this transaction, they deposit funds into the bribery contract and specify the required protocol deviation -- typically a structured message that the bribee must sign (\cf~\Cref{fig:class_structures}).
    \item[Take bribe] To accept the offer, the bribee calls the $\mathsf{takeBribe}()$ function. This call must include a proof of the required deviation, which is usually a cryptographic signature (\eg BLS or ECDSA) over a specific protocol message, such as an attestation or a voluntary exit transaction.
    \item[Check proof and pay out] Upon receiving the $\mathsf{takeBribe}()$ call, the contract first verifies the submitted proof (via the internal $\mathsf{checkProof()}$ function). If the proof is valid, the contract atomically pays the bribe to the bribee (via $\mathsf{payOut}()$ step). While these actions are often executed within the single $\mathsf{takeBribe}()$ transaction, they can be separated. For instance, in recurring bribes designed to maintain a specific behavior (\cf~\Cref{sec:pay_to_exit}), the proof submission and payout may be handled in distinct transactions.
\end{description}
\section{Bribing Contracts: \emph{``Bribers, Bribers On The Chain!''}}\label{sec:bribing_contracts}

\subsection{The $\mathsf{PayToAttest}$ bribery contract}\label{sec:pay_to_fork}

The $\mathsf{PayToAttest}$ contract enables a briber to purchase BLS signatures from a set of validators, $\{\pk_{i}\}^{n}_{i=1}$, for a specific block header $m$. To initiate the bribe, the briber commits the aggregate public key, $\pk=\prod\limits^{n}_{i=1} \pk_i$, and the block hash, $H(m)$, to the smart contract. Participating validators can then claim their reward by providing a valid aggregate BLS signature $\sigma$ that satisfies the verification equation in \Cref{eq:bls_same_message_check}. To ensure atomicity between the execution and beacon chains, EIP-4788 (\cf~\Cref{sec:pos_ethereum_prelims}) can be used to verify that this aggregate signature $\sigma$ was included in a beacon block.  %\bence{Megemlítsük, hogy ez arra is felhasználható, hogy incentivizálni lehessen a bribert, hogy ne cenzurázzon?}\istvan{absolutely!}

This mechanism is a powerful tool for facilitating forking attacks. For example, an attacker can orchestrate an ex-ante reorg~\cite{schwarz2022three} to manipulate the RANDAO beacon~\cite{nagy2025forking} or to steal MEV~\cite{DBLP:conf/eurosp/SarencheTMSP25}, as we detail in \Cref{sec:general_exante_reorgs}. While ex-ante reorgs are a known risk, it is widely believed that non-majority attackers cannot reliably launch ex-post reorgs~\cite{schwarz2022three}. However, we now describe how the $\mathsf{PayToAttest}$ contract provides a novel mechanism for a \emph{non-majority} briber to successfully perform such an attack.

\begin{figure}[t!]
\begin{center}
\begin{tikzpicture}[font=\small]

% Horizontal y level for main chain
\def\y{0}
\def\x{0}

% Slot dividing vertical lines (four slots)
\foreach \i in {0,1,2} {
  \draw[slotline] (2.5*\i,\y-1.8) -- (2.5*\i,\y+2.5);
}

% Attestors
\node[inner sep=0pt] (b1) at (.5,2.5)
    {\includegraphics[width=0.06\textwidth]{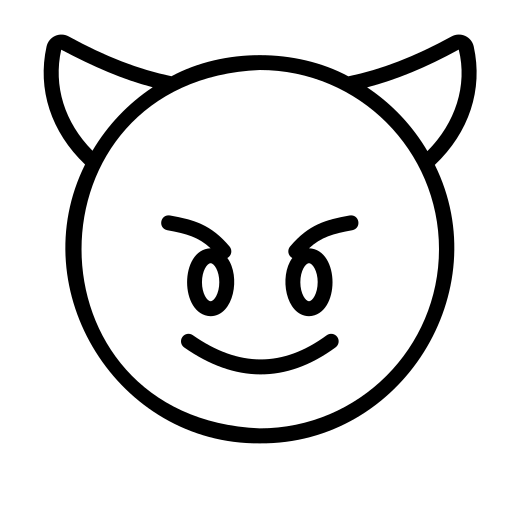}};
\node[inner sep=0pt] (a1) at (1.25,2.5)
    {\includegraphics[width=0.06\textwidth]{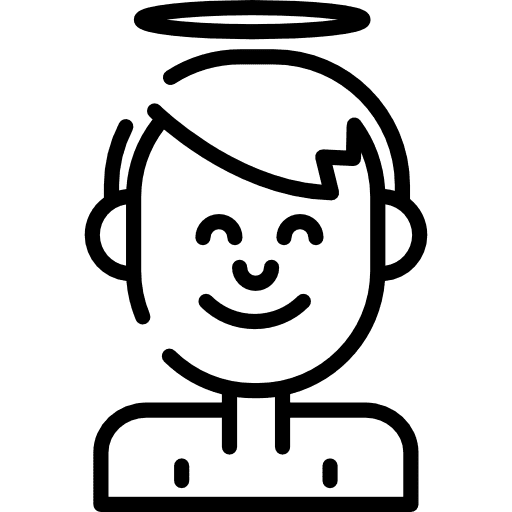}};
\node[inner sep=0pt] (r1) at (2.0,2.5)
    {\includegraphics[width=0.06\textwidth]{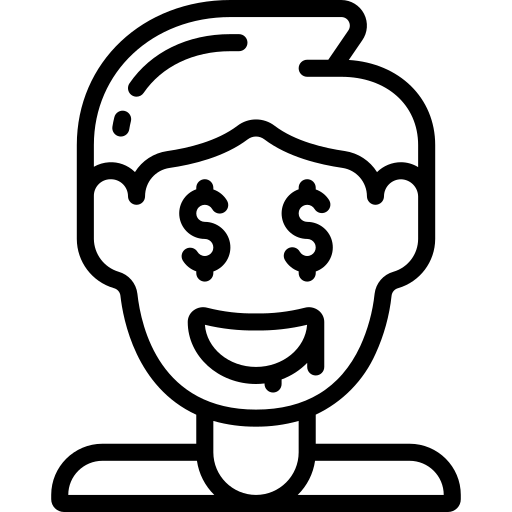}};
    
\node[inner sep=0pt] (b2) at (3,2.5)
    {\includegraphics[width=0.06\textwidth]{Figures/b.png}};
\node[inner sep=0pt] (a2) at (3.75,2.5)
    {\includegraphics[width=0.06\textwidth]{Figures/a.png}};
\node[inner sep=0pt] (r2) at (4.5,2.5)
    {\includegraphics[width=0.06\textwidth]{Figures/r.png}};
    
\node[inner sep=0pt] (b3) at (5.5,2.5)
    {\includegraphics[width=0.06\textwidth]{Figures/b.png}};
\node[inner sep=0pt] (a3) at (6.25,2.5)
    {\includegraphics[width=0.06\textwidth]{Figures/a.png}};
\node[inner sep=0pt] (r3) at (7,2.5)
    {\includegraphics[width=0.06\textwidth]{Figures/r.png}};

%Blocks
\node[honest]    (k)      at (-1.25,1)            {$k$};
\node[honest] (k1)    at (1.25,1)       {$k{+}1$};
\node[adversary]    (k1a)     at (3.75,-0.2)            {$(k{+}1)^*$};
\node[adversary] (k2a)    at (6.25,-0.2)       {$(k{+}2)^*$};

% Attestation arrows
\draw[honestarw] (a1.south) to (k1);
\draw[advarw] (b1.south) to (k);
\draw[altarw] (r1.south) to (k);

\draw[honestarw] (a2.south) to (k1);
\draw[advarw] (b2.south) to (k1a);
\draw[altarw] (r2.south) to (k1a);

\draw[honestarw] (a3.south) to (k2a);
\draw[advarw] (b3.south) to (k2a);
\draw[altarw] (r3.south) to (k2a);

%Chain arrows
\draw[black, thick] (-2.5,1) -- (k);
\draw[achain] (k1a) -- (-1.25,-0.2) -- (k);
\draw[achain] (k2a) -- (k1a);
\draw[chain] (k1) -- (k);
\draw[black, thick] (k1.east) -- (3.5,1);
\draw[black, thick] (k1.east) -- (3.75,1);
\draw[black, thick, dashed] (3.75,1) -- (7.2,1);

% Slot labels under diagram
\node[align=center] at (-1.25,\y-1.3)  {Slot $n$ \\ \textcolor{honestblue}{\textbf{H}onest}};
\node[align=center] at (1.25,\y-1.3)  {Slot $n{+}1$ \\ \textcolor{honestblue}{\textbf{H}onest}};
\node[align=center] at (3.75,\y-1.3)  {Slot $n{+}2$ \\ \textcolor{adversaryred}{\textbf{A}dversary}};
\node[align=center] at (6.25,\y-1.3)  {Slot $n{+}3$ \\ \textcolor{adversaryred}{\textbf{A}dversary}};

\node[align=center] at (8.5,1) {$2(1-\alpha)\cdot$\\$\cdot(1-\beta)$};
\node[align=center] at (8.5,\y-0.2) {$0.4 + \alpha {+} $\\$ +(1-\alpha)\beta$};
\node[align=center] at (8.5,\y-1.3) {Vote\\ weights};

\end{tikzpicture}
\end{center}
\caption{Leveraging the $\mathsf{PayToAttest}$ bribery contract to ex-post reorg the block proposed in Slot ($n+1$). Colored arrows indicate which blocks different validators vote for as the head of the blockchain. Black arrows represent hash pointers. Red (blue) blocks are proposed by the briber (honest validators).}
\label{fig:ex_post_fork_explainer}
\end{figure}

%OLD
%The ex-post reorg proceeds as follows. For concreteness, let us assume that two honest block proposers are followed by two slots in which the briber can propose blocks,~\cf~\Cref{fig:ex_post_fork_explainer}.\footnote{We study ex-post reorgs for general chain strings in~\Cref{sec:exante_paytofork}.} The adversary wants to fork out Block $k+1$ in Slot $n+1$ (perhaps due to its high MEV content). The adversary bribes the economically rational attestators to vote for Block $k$ even though Block $k+1$ has been published. Honest validators vote for Block $k+1$ as the head of the blockchain both in Slots $n+1$ and $n+2$. The adversary bribes again rational attestators to vote for its Block $(k+1)^*$. Whenever the adversary publishes its Block $(k+2)^*$, the adversary's fork accumulated $\pboost{} + \alpha + (1-\alpha)\beta$ attestations. The honest validators' fork has $2(1-\alpha)(1-\beta)$ vote.  In summary, if ~\cf~\Cref{eq:expost_winning_condition} holds, then the adversary's fork is considered the canonical chain. Finally, all validators vote for Block $(k+2)^*$ as the head of the blockchain in Slot $n+3$ completing the attack. A complete sequence diagram of the ex-post reorg involving the briber, bribee, honest validators, and the $\mathsf{PayToAttest}$ contract is shown in~\Cref{fig:pay_to_attest_ex_post_timeline}. The bribery costs for ex-post forking attacks are studied in~\Cref{sec:pay_to_fork_incentives}.

The ex-post reorg proceeds as follows. For concreteness, we assume a scenario where an honest block is published in Slot $n+1$, followed by two slots where the briber is the proposer, as depicted in~\Cref{fig:ex_post_fork_explainer}. The adversary’s goal is to fork out Block $k+1$, perhaps for its high MEV content. After Block $k+1$ is published, the adversary bribes rational attestators to instead vote for the preceding Block $k$. In the next slot, the adversary proposes their own block, Block $(k+1)^*$, which builds on Block $k$, and again bribes attestators to vote for it, while honest validators continue to vote for Block $k+1$. When the adversary publishes their second block, Block $(k+2)^*$, their fork has accumulated a weight of $\pboost{} + \alpha + (1-\alpha)\beta$ attestations, compared to the honest fork's $2(1-\alpha)(1-\beta)$ votes. If the following condition holds,
\begin{equation}\label{eq:expost_winning_condition}
    2(1-\alpha)(1-\beta)\leq \pboost{} + \alpha + (1-\alpha)\beta\enspace , 
\end{equation}
then the adversary's fork becomes canonical, and all validators then vote for Block $(k+2)^*$ in Slot $n+3$, completing the attack. A complete sequence diagram is shown in~\Cref{fig:pay_to_attest_ex_post_timeline}, and the associated bribery costs are studied in~\Cref{sec:pay_to_fork_incentives}.

\subsection{The $\mathsf{PayToExit}$ bribery contract}\label{sec:pay_to_exit}

The $\mathsf{PayToExit}$ contract creates a market where a briber can incentivize validators to voluntarily exit the Ethereum consensus protocol. A briber's motivation could be multifold: they can harm the network's liveness by reducing the number of active participants, increase their own relative staking power, or attract scarce capital from the deposit contract to an on-chain application (\eg a decentralized exchange or lending platform). 

To claim the reward, a validator must prove to the contract that they have successfully exited the validator set after the bribe was offered. This proof consists of two components: a Merkle proof confirming their prior status as an active validator, and a valid BLS signature on a \texttt{VoluntaryExit} message (\cf~\Cref{fig:class_structures}). Note that such a trustless $\mathsf{PayToExit}$ contract is not feasible in Proof-of-Work, as there is no on-chain mechanism to verify that a miner has permanently shut down their hardware. The pseudocode for this contract is detailed in \Cref{alg:paytoexit}.

A key challenge in this model is the restaking loophole. A rational validator could exit, claim the bribe, transfer their funds to a fresh address, and rejoin the validator set. A more robust design could mitigate this by transforming the bribe into a recurring payment, rewarding the bribee as long as their withdrawn funds remain at the withdraw address, thus disincentivizing restaking. We leave the implementation of this extension of our $\mathsf{PayToExit}$ contract to future work.

%time-locking mechanism for the pay-to-exit contract to ensure that the bribee does not reuse its capital. We leave this extension to future work.

\subsection{The $\mathsf{PayToBias}$ bribery contract}\label{sec:RANDAO_bribery_market}

Ethereum's RANDAO is a distributed randomness beacon used to select block proposers. Its output for an epoch $e$ is calculated as $R^{e} := R^{e-1} \oplus (\oplus^{32}_{i=1}r^{e}_{i})$, where $r^{e}_i$ is the proposer's randomness contribution for each slot. Since each contribution $r^{e}_i$ is a BLS signature ($\mathsf{BLS.Sign}(\mathsf{sk}_i,e)$), it is unique and cannot be grind. The only influence a validator has is the binary choice to publish their block (revealing their contribution) or withhold it. This design is known to be biasable, as it gives proposers in the final ``tail'' slots of an epoch significant power to influence the outcome, which in turn determines the proposers for epoch $e+2$~\cite{alpturer2024optimal,nagy2025forking,tapolcai2025slot,wahrstatter2023selfish}.

A rational validator controlling $k$ tail slots can therefore pre-compute $2^k$ possible RANDAO outcomes and select the one most beneficial to themselves. Alternatively, they can auction this right to influence the outcome to other parties. This concept is known as the RANDAO bribery market, and our $\mathsf{PayToBias}$ contract is, to our knowledge, its first practical implementation.

The auction mechanism is straightforward. The manipulator (the validator controlling $k$ tail slots) first calls $\mathsf{offerBribe()}$ to reveal their $k$ randomness contributions,~\ie $\forall i\in[k]:r^e_{32-i}$. Once the epoch's other contributions are public, other validators bid on their preferred outcome by submitting a publish/withhold configuration ($\mathbf{c}_i\in\{0,1\}^{k}$). The manipulator is then incentivized to execute the configuration $\mathbf{c}_i$ that attracted the highest total value of bids.

Finally, to claim their payout, the manipulator must prove which configuration they executed by submitting the relevant block headers from the epoch's conclusion. The contract verifies the authenticity of these headers against the on-chain ``blockhash'' history, which is available for up to $8,191$ blocks ($\approx$27.3 hours) per EIP-2935~\cite{vitalik2020eip}. By inspecting the timestamps of the verified headers, the contract confirms which slots were missed, determines the executed configuration $\mathbf{c}_i$, and releases the funds. The complete pseudocode \revision{of the RANDAO bribery market} is available in \Cref{sec:paytobias_implementation}. \revision{In~\Cref{sec:pay_to_bias_incentives}, we compute a worst-case expected upper bound for the bribe amount to manipulate the RANDAO, whenever the manipulator has exactly $k$ tail slots.}
\subsection{Performance Evaluation and Qualitative Comparison}\label{sec:evaluation}

%------------------------------Begin Gas Performance Table--------------------------
\begin{table*}[t]
\centering
\begin{minipage}{\textwidth}
\begin{center}
 \resizebox{\textwidth}{!}{
 \begin{tabular}{l c c c} 
 \toprule
   \diagbox[]{Contract}{Function}&$\mathsf{Constructor(\cdot)}$&$\mathsf{offerBribe(\cdot)}$&$\mathsf{takeBribe(\cdot)}$ \\ [0.5ex] 
\midrule
$\mathsf{PayToAttest}$ & $805,876$ ($4.87$ USD) &$361,117$ ($2.18$ USD) & $252,110$ ($1.52$ USD) \\

$\mathsf{PayToExit}$ & $862,575$ ($5.21$ USD) &$225,123$ ($1.36$ USD) & $260,643$ ($1.58$ USD) \\

$\mathsf{PayToBias}$ & $1,296,940$ ($7.84$ USD) & $94,944$ ($0.57$ USD) & $137,897$ ($0.83$ USD) \\
\bottomrule
\end{tabular}
}
\caption{The most important functions' gas costs in our bribing contracts. USD costs were computed using ETH/USD exchange price and gas prices ($3,708$ USD/ETH and  $1.63$ Gwei gas price) on July 25th, 2025.}
\label{table:practical_performance_gas_costs}
\end{center}
\end{minipage}
\end{table*}
%------------------------------END Gas Performance Table---------------------------

We implemented our proposed bribery contracts in Solidity and evaluated the gas costs incurred by both the briber and the bribees. Our evaluation confirms that the contracts are highly efficient, imposing negligible on-chain financial costs for participants in these bribery markets (\cf~\Cref{table:practical_performance_gas_costs}). A qualitative comparison of our contracts with related work on bribery attacks is provided in \Cref{tab:attacks_qualitative_comparison}. The complete implementation is publicly available in our open-source repository.\footnote{\ifanonymous\url{https://anonymous.4open.science/r/bribery-zoo-70FD/}\else\url{https://github.com/0xSooki/bribery-zoo}\fi}

%We implemented (\ifanonymous\url{https://anonymous.4open.science/r/bribery-zoo-70FD/}\else\url{https://github.com/0xSooki/bribery-zoo}\fi) our proposed bribery contracts in Solidity and evaluated their incurred gas costs both for the briber and the bribee. All our bribery contracts are very efficient and, in practice, they incur essentially negligible financial costs for all participants to participate in the respective bribery markets,~\cf~\Cref{table:practical_performance_gas_costs}. We provide a qualitative comparison of our introduced bribery contracts with related works on bribery attacks in~\Cref{tab:attacks_qualitative_comparison}.

\subsection{Privacy-preserving Bribery Contracts}\label{sec:privacy_preserving_bribery}

The public nature of the blockchain may deter validators from participating in bribery markets, as their identifying details (e.g., public keys, validator indices) and transaction amounts are transparent. A natural extension of our work is therefore to incorporate privacy protections into these contracts. The core computations in our contracts, such as verifying Merkle proofs and BLS signatures, are well-suited for zero-knowledge techniques. These operations could be embedded within a zkSNARK, such as the widely-used Groth16 proof system~\cite{groth2016size}, which is particularly efficient for on-chain verification. While we leave the implementation of a fully private bribery contract to future work, this represents a promising direction for making such bribery markets more practical.

\revision{
\subsection{Responsible disclosure and ethical considerations}
We responsibly disclosed our findings to the Ethereum Foundation before publishing our work. We ran and tested our bribery contract implementations solely on a local test blockchain without interfering with the live Ethereum consensus.
}

%Validators might be reluctant to participate in a bribery market if their personal details (public keys, validator indices, bribe amounts) are transparent on a public blockchain. Hence, a natural extension of the above bribery contracts is to endow them with privacy protections. Our bribery contracts check Merkle membership proofs, verify BLS signatures, etc. These simple computations can be made zero-knowledge using general techniques such as zkSNARKs,~\eg with the popular proof system of Groth~\cite{groth2016size} that has a particularly on-chain friendly efficient verifier. We leave the implementation of this extension to future work.
\section{Bribee Incentives: \emph{``Is Resisting All in Vain?''}}\label{sec:incentives}

\subsection{$\mathsf{PayToAttest}$ incentives}\label{sec:pay_to_fork_incentives}

As established in~\Cref{sec:pay_to_fork}, a non-majority adversary can perform an ex-post reorg using a $\mathsf{PayToAttest}$ bribery contract to fork out an honest block $\honestblock{n+1}{e}$ at Slot $k+1$ for a chain string $\honestblock{n}{e}\honestblock{n+1}{e}\advblock{n+2}{e}\advblock{n+3}{e}$ within epoch $e$. According to the winning condition in~\Cref{eq:expost_winning_condition}, the briber must purchase attestations from a $(1-\alpha)\beta$ fraction of the $\frac{N(t)}{32}$ validators in the attestation committee. The protocol rewards each correct and timely attestation with an amount proportional to the total stake, approximately $c \cdot \sqrt{S(t)}$ ETH. Therefore, to succeed, the briber only needs to offer a total bribe slightly greater than the total rewards these validators would otherwise receive (\ie $(1-\alpha)\beta \frac{N(t)}{32}(c\sqrt{S(t)}$)).

Our analysis of the concrete costs, using network data from April 1, 2025, is presented in~\Cref{fig:pay_to_attest_ex_post_costs}. The results are alarming: for all successful attack parameters $(\alpha, \beta)$, the total bribe required is less than $0.09$ ETH ($\approx$ $334$ USD). This cost is significantly lower than the typical MEV found in a single Ethereum block~\cite{wahrstatter2023time}, implying that such an attack could be highly profitable.

%Recall from~\Cref{sec:pay_to_fork}, that for a chain string $\honestblock{n}{e}\honestblock{n+1}{e}\advblock{n+2}{e}\advblock{n+3}{e}$ in epoch $e$, a briber can fork out the honest block in Slot $k+1$. Specifically,~\Cref{eq:expost_winning_condition} tells us that the briber must buy $(1-\alpha)\beta$ fraction of the $\frac{N(t)}{32}$ attestations to succeed in its ex-post reorg attempt. The Ethereum PoS protocol pays $c\sqrt{S(t)}$ ether for each timely and correct attestation to any validator, where $c$ is a constant, and $S(t)$ denotes the total ether staked. 
%Typically, the briber must pay an extremely low price (\ie $(1-\alpha)\beta \frac{N(t)}{32}(c\sqrt{S(t)}$)) to ex-post reorg successfully the blockchain. For concrete values as of April 1, 2025,~\cf~\Cref{fig:pay_to_fork_bribe_amounts}. The bribe costs obtained in~\Cref{fig:pay_to_fork_bribe_amounts} are particularly alarming, since the MEV content of a typical block (as measured in~\cite{wahrstatter2023time}) is larger than the $\mathsf{PayToAttest}$  bribe cost to fork the blockchain. We find in~\Cref{fig:pay_to_fork_bribe_amounts} that for all successful ex-post reorg parameters $(\alpha,\beta)$, the briber needs to pay less than $0.09$ ether ($\approx334$ USD).

\begin{figure}[t!]
    \centering
    \subfloat[Ex-post reorg costs (ETH) using the $\mathsf{PayToAttest}$ contract for the chainstring $\honestblock{n}{e}\honestblock{n+1}{e}\advblock{n+2}{e}\advblock{n+3}{e}$,~\cf~\Cref{sec:pay_to_fork}. The attack is unsuccessful in the white region.]{{
    \resizebox{0.49\textwidth}{!}{% This file was created with matplot2tikz v0.4.0.
\begin{tikzpicture}

\definecolor{darkgray176}{RGB}{176,176,176}

\begin{axis}[
colorbar,
colorbar style={ylabel={\textbf{Bribe amount (ETH)}}},
colormap={mymap}{[1pt]
  rgb(0pt)=(1,1,0.8);
  rgb(1pt)=(1,0.929411764705882,0.627450980392157);
  rgb(2pt)=(0.996078431372549,0.850980392156863,0.462745098039216);
  rgb(3pt)=(0.996078431372549,0.698039215686274,0.298039215686275);
  rgb(4pt)=(0.992156862745098,0.552941176470588,0.235294117647059);
  rgb(5pt)=(0.988235294117647,0.305882352941176,0.164705882352941);
  rgb(6pt)=(0.890196078431372,0.101960784313725,0.109803921568627);
  rgb(7pt)=(0.741176470588235,0,0.149019607843137);
  rgb(8pt)=(0.501960784313725,0,0.149019607843137)
},
point meta max=0.0803152693076634,
point meta min=0.00286810981100726,
tick align=outside,
tick pos=left,
x grid style={darkgray176},
xlabel={\(\displaystyle \boldsymbol{\alpha}\)},
xmin=0, xmax=100,
xtick style={color=black},
xtick={0,9,19,29,39,49,59,69,79,89,99},
xticklabels={0.01,0.05,0.1,0.15,0.2,0.25,0.3,0.35,0.4,0.45,0.5},
xticklabel style={rotate=45.0},
y grid style={darkgray176},
ylabel={\(\displaystyle \boldsymbol{\beta}\)},
ymin=0, ymax=100,
ytick style={color=black},
ytick={0,9,19,29,39,49,59,69,79,89,99},
yticklabels={0.0,0.09,0.19,0.29,0.39,0.49,0.6,0.7,0.8,0.9,1.0}
]
\addplot graphics [includegraphics cmd=\pgfimage,xmin=0, xmax=100, ymin=0, ymax=100] {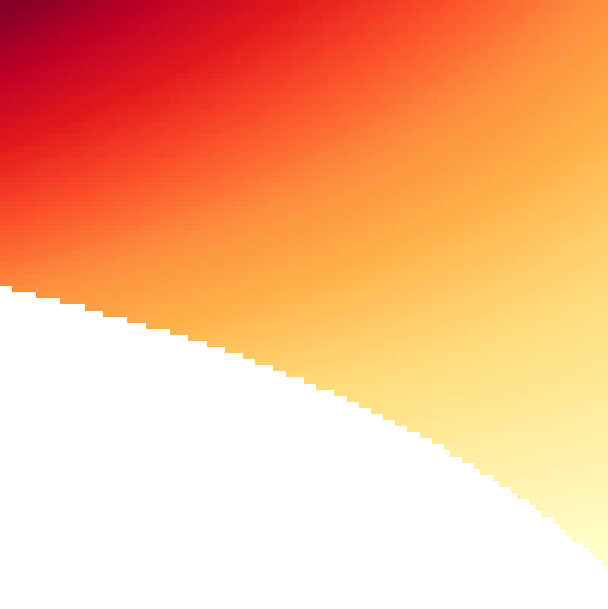};
\draw (axis cs:93.5,16.5) node[
  scale=0.75,
  text=black,
  rotate=0.0
]{\bfseries 0.01};
\draw (axis cs:71.5,27.5) node[
  scale=0.75,
  text=black,
  rotate=0.0
]{\bfseries 0.01};
\draw (axis cs:82.5,27.5) node[
  scale=0.75,
  text=black,
  rotate=0.0
]{\bfseries 0.01};
\draw (axis cs:93.5,27.5) node[
  scale=0.75,
  text=black,
  rotate=0.0
]{\bfseries 0.01};
\draw (axis cs:49.5,38.5) node[
  scale=0.75,
  text=black,
  rotate=0.0
]{\bfseries 0.02};
\draw (axis cs:60.5,38.5) node[
  scale=0.75,
  text=black,
  rotate=0.0
]{\bfseries 0.02};
\draw (axis cs:71.5,38.5) node[
  scale=0.75,
  text=black,
  rotate=0.0
]{\bfseries 0.02};
\draw (axis cs:82.5,38.5) node[
  scale=0.75,
  text=black,
  rotate=0.0
]{\bfseries 0.02};
\draw (axis cs:93.5,38.5) node[
  scale=0.75,
  text=black,
  rotate=0.0
]{\bfseries 0.02};
\draw (axis cs:16.5,49.5) node[
  scale=0.75,
  text=black,
  rotate=0.0
]{\bfseries 0.04};
\draw (axis cs:27.5,49.5) node[
  scale=0.75,
  text=black,
  rotate=0.0
]{\bfseries 0.03};
\draw (axis cs:38.5,49.5) node[
  scale=0.75,
  text=black,
  rotate=0.0
]{\bfseries 0.03};
\draw (axis cs:49.5,49.5) node[
  scale=0.75,
  text=black,
  rotate=0.0
]{\bfseries 0.03};
\draw (axis cs:60.5,49.5) node[
  scale=0.75,
  text=black,
  rotate=0.0
]{\bfseries 0.03};
\draw (axis cs:71.5,49.5) node[
  scale=0.75,
  text=black,
  rotate=0.0
]{\bfseries 0.03};
\draw (axis cs:82.5,49.5) node[
  scale=0.75,
  text=black,
  rotate=0.0
]{\bfseries 0.02};
\draw (axis cs:93.5,49.5) node[
  scale=0.75,
  text=black,
  rotate=0.0
]{\bfseries 0.02};
\draw (axis cs:5.5,60.5) node[
  scale=0.75,
  text=black,
  rotate=0.0
]{\bfseries 0.05};
\draw (axis cs:16.5,60.5) node[
  scale=0.75,
  text=black,
  rotate=0.0
]{\bfseries 0.04};
\draw (axis cs:27.5,60.5) node[
  scale=0.75,
  text=black,
  rotate=0.0
]{\bfseries 0.04};
\draw (axis cs:38.5,60.5) node[
  scale=0.75,
  text=black,
  rotate=0.0
]{\bfseries 0.04};
\draw (axis cs:49.5,60.5) node[
  scale=0.75,
  text=black,
  rotate=0.0
]{\bfseries 0.04};
\draw (axis cs:60.5,60.5) node[
  scale=0.75,
  text=black,
  rotate=0.0
]{\bfseries 0.03};
\draw (axis cs:71.5,60.5) node[
  scale=0.75,
  text=black,
  rotate=0.0
]{\bfseries 0.03};
\draw (axis cs:82.5,60.5) node[
  scale=0.75,
  text=black,
  rotate=0.0
]{\bfseries 0.03};
\draw (axis cs:93.5,60.5) node[
  scale=0.75,
  text=black,
  rotate=0.0
]{\bfseries 0.03};
\draw (axis cs:5.5,71.5) node[
  scale=0.75,
  text=black,
  rotate=0.0
]{\bfseries 0.06};
\draw (axis cs:16.5,71.5) node[
  scale=0.75,
  text=black,
  rotate=0.0
]{\bfseries 0.05};
\draw (axis cs:27.5,71.5) node[
  scale=0.75,
  text=black,
  rotate=0.0
]{\bfseries 0.05};
\draw (axis cs:38.5,71.5) node[
  scale=0.75,
  text=black,
  rotate=0.0
]{\bfseries 0.05};
\draw (axis cs:49.5,71.5) node[
  scale=0.75,
  text=black,
  rotate=0.0
]{\bfseries 0.04};
\draw (axis cs:60.5,71.5) node[
  scale=0.75,
  text=black,
  rotate=0.0
]{\bfseries 0.04};
\draw (axis cs:71.5,71.5) node[
  scale=0.75,
  text=black,
  rotate=0.0
]{\bfseries 0.04};
\draw (axis cs:82.5,71.5) node[
  scale=0.75,
  text=black,
  rotate=0.0
]{\bfseries 0.03};
\draw (axis cs:93.5,71.5) node[
  scale=0.75,
  text=black,
  rotate=0.0
]{\bfseries 0.03};
\draw (axis cs:5.5,82.5) node[
  scale=0.75,
  text=black,
  rotate=0.0
]{\bfseries 0.06};
\draw (axis cs:16.5,82.5) node[
  scale=0.75,
  text=black,
  rotate=0.0
]{\bfseries 0.06};
\draw (axis cs:27.5,82.5) node[
  scale=0.75,
  text=black,
  rotate=0.0
]{\bfseries 0.06};
\draw (axis cs:38.5,82.5) node[
  scale=0.75,
  text=black,
  rotate=0.0
]{\bfseries 0.05};
\draw (axis cs:49.5,82.5) node[
  scale=0.75,
  text=black,
  rotate=0.0
]{\bfseries 0.05};
\draw (axis cs:60.5,82.5) node[
  scale=0.75,
  text=black,
  rotate=0.0
]{\bfseries 0.05};
\draw (axis cs:71.5,82.5) node[
  scale=0.75,
  text=black,
  rotate=0.0
]{\bfseries 0.04};
\draw (axis cs:82.5,82.5) node[
  scale=0.75,
  text=black,
  rotate=0.0
]{\bfseries 0.04};
\draw (axis cs:93.5,82.5) node[
  scale=0.75,
  text=black,
  rotate=0.0
]{\bfseries 0.04};
\draw (axis cs:5.5,93.5) node[
  scale=0.75,
  text=black,
  rotate=0.0
]{\bfseries 0.07};
\draw (axis cs:16.5,93.5) node[
  scale=0.75,
  text=black,
  rotate=0.0
]{\bfseries 0.07};
\draw (axis cs:27.5,93.5) node[
  scale=0.75,
  text=black,
  rotate=0.0
]{\bfseries 0.07};
\draw (axis cs:38.5,93.5) node[
  scale=0.75,
  text=black,
  rotate=0.0
]{\bfseries 0.06};
\draw (axis cs:49.5,93.5) node[
  scale=0.75,
  text=black,
  rotate=0.0
]{\bfseries 0.06};
\draw (axis cs:60.5,93.5) node[
  scale=0.75,
  text=black,
  rotate=0.0
]{\bfseries 0.05};
\draw (axis cs:71.5,93.5) node[
  scale=0.75,
  text=black,
  rotate=0.0
]{\bfseries 0.05};
\draw (axis cs:82.5,93.5) node[
  scale=0.75,
  text=black,
  rotate=0.0
]{\bfseries 0.04};
\draw (axis cs:93.5,93.5) node[
  scale=0.75,
  text=black,
  rotate=0.0
]{\bfseries 0.04};
\end{axis}

\end{tikzpicture}}
    
    }\label{fig:pay_to_attest_ex_post_costs}}%
    \qquad
    \subfloat[$\mathsf{PayToExit}$ bribery contract costs (ETH) in the number of bribed validators for different $(Y,r)$ parameters (year, discount rate).]{{
    \resizebox{0.41\textwidth}{!}{\input{Figures/payToExitBribeAmounts}}
    }}%
    \caption{$\mathsf{PayToExit}$ and $\mathsf{PayToAttest}$ (ex-post reorg) bribery contract costs.}%
    \label{fig:pay_to_fork_bribe_amounts}%
\end{figure}

\subsection{$\mathsf{PayToExit}$ incentives}\label{sec:pay_to_exit_incentives}

We analyze the $\mathsf{PayToExit}$ bribery market through a game-theoretic lens, modeling it as a two-stage Stackelberg market exit game~\cite{stackelberg1934marktform} with a single leader (the briber) and multiple followers (the validators). First, the leader commits to a bribe $b$; second, each follower decides whether to accept the bribe and exit, or to refuse and remain. The objective of this analysis is to demonstrate the practicality and economic rationality behind this attack. To this end, we analyze a simplified game that captures these essential dynamics, as illustrated in~\Cref{fig:paytoexit_game}.

%Next, we study the $\mathsf{PayToExit}$ bribery market through a game-theoretic lens. We define the $\mathsf{PayToExit}$ game as a single-leader (briber) multiple-follower (bribees) Stackelberg market exit game~\cite{stackelberg1934marktform}. %First, the leader decides on a bribe $b$ and announces it via a smart contract. Second, the followers decide individually whether they a) take the bribe and exit the market, or b) refuse the bribe and stay in the market. The game mechanism is illustrated in Figure~\ref{fig:paytoexit_game}. Note that our objective with this analysis is to demonstrate the rationality and practicality of the $\mathsf{PayToExit}$ bribery attack. To this end, we define and analyze the simplest sensible game that captures such aspects.

\paragraph{Assumptions. }Our analysis is based on several simplifying assumptions. We assume the briber's (leader's) objective is to achieve a target market share of $\alpha^* > \alpha$, which requires a total of $k^*$ validators to exit. We assume the following.

\begin{description}
    \item[Homogeneous Validators] The follower set is homogeneous, \ie all rational validators have identical utility functions and a uniform stake ($=32$ ETH).
    \item[Sufficient Rational Validators] The network has $N$ total validators, where the number of rational validators is $n=(1-\alpha)\beta N$. This value is strictly larger than the required number of exits $k^*$.
    \item[Constant Bribe] A single, constant bribe $b$ is offered. This represents an upper bound for the required bribe, as it is the amount needed to convince the marginal (last) validator to exit,~\cf~\Cref{sec:apr-rect}.
    %\item[Minimum Gain Threshold] Rational validators are staying unless the utility from exiting exceeds the utility from staying by at least $\epsilon > 0$.
    \item[Perfect Information] The game has perfect information, as the $\mathsf{PayToExit}$ smart contract makes all relevant data -- such as the number of validators who have already committed to exiting -- publicly visible to all participants.
    \item[Simplified Exits] We do not model the sequential nature of the exit process, \ie queueing delays or specific timestamps; our follower game is single-shot.
\end{description}

%\paragraph{Assumptions. } We assume that the leader's objective is to achieve a $\alpha^* > \alpha$ market share that requires $k^*$ exiting validators. For simplicity, we assume:
%\begin{itemize}
%    \item A homogeneous set of followers of a single ``type'' with identical utility functions and unit market share (corresponding to $32$ ETH);
%    \item A total number of validators $N$, and the number of rational validators $n$ is strictly larger than $k^*$ (the number of exited validators in equilibrium);
%    \item A constant bribe $b$ -- this yields an upper limit for the equilibrium bribe, as it corresponds to the last exiting validator. Earlier exits can be motivated with lower bribes -- the bribery smart contract could handle dynamic payouts;
%    \item A game of perfect information -- the $\mathsf{PayToExit}$ bribery smart contract provides all necessary information to the players, including how many exiters have committed to the contract at any given time;
%    \item We do not explicitly model the sequential nature of exits, including exit timestamps and queueing delays for actually leaving the system.
%\end{itemize}

\begin{figure}[tb]
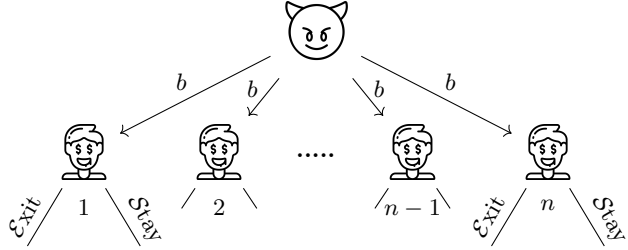

\centering
\begin{tikzpicture}[scale=0.5][
    circ/.style={circle, draw, minimum size=6mm, inner sep=0pt},
    smallcirc/.style={--, draw, minimum size=3mm, inner sep=0pt},
    >=Stealth
]

% Top node (Briber)
\node[] (briber) {\includegraphics[width=0.08\textwidth]{Figures/b.png}};

% Four children aligned horizontally
\node[below left=0.5cm and 2cm of briber] (c1) {\includegraphics[width=0.06\textwidth]{Figures/r.png}};
\node[below left=0.5cm and .2cm of briber] (c2) {\includegraphics[width=0.06\textwidth]{Figures/r.png}};
\node[below right=0.5cm and .2cm of briber] (c3) {\includegraphics[width=0.06\textwidth]{Figures/r.png}};
\node[below right=0.5cm and 2cm of briber] (c4) {\includegraphics[width=0.06\textwidth]{Figures/r.png}};
\node[below=25pt of briber] (ccc) {\textbf{.....}};

\node[below= 1pt of c1] {$1$};
\node[below= 1pt of c2] {$2$};
\node[below= 1pt of c3] {$n-1$};
\node[below= 1pt of c4] {$n$};

% Strategies
\draw (c1) -- ++(-1.5,-2.5) node[midway, left, rotate=60, xshift=8pt, yshift=7pt] {$\mathcal{E}$xit};
\draw (c1) -- ++(1.5,-2.5) node[midway, right, rotate=-60, xshift=-8pt, yshift=7pt] {$\mathcal{S}$tay};
\draw (c2) -- ++(-1,-1.5);
\draw (c2) -- ++(1,-1.5);
\draw (c3) -- ++(-1,-1.5);
\draw (c3) -- ++(1,-1.5);
\draw (c4) -- ++(-1.5,-2.5) node[midway, left, rotate=60, xshift=8pt, yshift=7pt] {$\mathcal{E}$xit};
\draw (c4) -- ++(1.5,-2.5) node[midway, right, rotate=-60, xshift=-8pt, yshift=7pt] {$\mathcal{S}$tay};

% Edges from Briber
\draw[->] (briber) -- (c1) node[midway, left, yshift=5pt] {$b$};
\draw[->] (briber) -- (c2) node[midway, left, xshift=1pt, yshift=3pt] {$b$};
\draw[->] (briber) -- (c3) node[midway, right, xshift=-1pt, yshift=3pt] {$b$};
\draw[->] (briber) -- (c4) node[midway, right, yshift=5pt] {$b$};

% LABELS
%\node[right=70pt of briber, align=center] {2) LEADER\\OPTIMIZATION};
%\node[below right=25pt and 70pt of briber, align=center] {1) FOLLOWER\\GAME};

% Dotted divider line
%\draw[dashed] ($(briber.south)+(-4.5,-0.2)$) -- ($(briber.south)+(5.5,-0.2)$);
%\draw[solid] ($(briber.south)+(-4.5,-3.2)$) -- ($(briber.south)+(5.5,-3.2)$) node[midway, below] {$\pi_0, \pi_1,..., \pi_n$};
\end{tikzpicture}
\caption{$\mathsf{PayToExit}$: a single-leader multiple-follower Stackelberg market exit game. A briber offers a bribe $b$ to $n$ rational validators to increase its market share $\alpha$.}
\label{fig:paytoexit_game}
\end{figure}

%\paragraph{Players. } A single briber acts as the leader; this briber wants to achieve a larger market share $\alpha^*$, requiring $k^*$ validators to exit the market. There are $n>k^*$ followers, all \emph{rational} validators, who aim to maximize their profit,~\cf~\Cref{fig:paytoexit_game}.

%\paragraph{Strategies/Action set. } The leader's strategy is $S_0 = \{b | b \in \mathbb{R^+}\}$, the bribe amount. For each follower, its strategy set is $S_i \in \{\mathcal{E}\mathrm{(xit})\}, \mathcal{S}\mathrm{(tay)}\}$, where $S_i = \mathcal{E}$ stands for follower $i$ accepting bribe $b$ and exiting, while $S_i = \mathcal{S}$ stands for follower $i$ refusing the bribe $b$ and staying in the Ethereum PoS consensus protocol.

\paragraph{Players and Strategies. } The game has two types of players,~\cf~\Cref{fig:paytoexit_game}. The leader is a single briber with $\alpha$ stake, whose objective is to increase their market share to $\alpha^*$, an outcome requiring $k^*$ validators to exit. The followers are the $n>k^*$ rational validators who aim to maximize their individual profit. The leader's strategy set, $S_0 = \{b | b \in \mathbb{R^+}\}$, is the choice of a bribe amount $b$. Each follower's strategy set is the binary choice $S_i \in \{\mathcal{E}, \mathcal{S}\}$, where $\mathcal{E}$ represents accepting the bribe and exiting the protocol, and $\mathcal{S}$ represents refusing the bribe and staying.

\paragraph{Leader's Utility.} The leader's (briber's) utility, $\Pi_0$, is defined in~\Cref{eq:leader_utility}. It consists of the opportunity gain from the increased market share, $g(k)$, plus the extra profit, $\Pi_0^e$, collected from sources exogenous to the consensus protocol, minus the total cost of the bribes, \ie $k \cdot b$ where $k$ is the number of validators who exit. Exogenous motivations of a $\mathsf{PayToExit}$ briber may be to disrupt Ethereum by creating liveness problems. Another, exogenous profit source of a $\mathsf{PayToExit}$ briber could be to attract scarce capital from the deposit contract to their own decentralized applications, akin to vampire attacks~\cite{hatfield2023simple}. The gain function $g(k)$ represents the net present value of the briber's increased future staking rewards. It is calculated by multiplying the annual reward gain, $\mathrm{R}(N-k)-\mathrm{R}(N)$, by a present value multiplier, $\mathrm{PV}(r,Y)$. More specifically, $\mathrm{R}(n)$ is the total annual staking reward for a single validator when $n$ validators are active: it consists of the attestation and block proposer rewards (sync committee rewards are not modeled)~\cite{pavloff2023ethereum}.
%and the mean MEV content of a block~\footnote{\url{https://dune.com/queries/1432505/2429301}}. 
We estimated the yearly average MEV amount per validator using on-chain data, \ie we computed the average payment value that block builders paid to Ethereum block proposers from September 15, 2022, to September 12, 2025.
The term $\mathrm{PV}(r,Y)$ is the present value multiplier, with a discount factor $r$ over a time horizon of $Y$ years as in~\cite{Brealey2020}. We set this discount rate to $r=0.08$: this value is higher than the return on risk-free assets like treasury bonds ($4-5\%$) to account for Ethereum's volatility, yet lower than returns on more speculative assets ($\ge10\%$), aligning with the crypto staking literature~\cite{cong2025tokenomics}. For the time horizon $Y$, we use an effective horizon based on the Present Value half-life~\cite{Brealey2020}, which is approximately $9$ years in this case (\cf~\Cref{sec:halflife}). %Finally, the annual reward $R$ per validator is closely estimated by summing the expected annual attestation and block rewards.

%\paragraph{Utilities. } The leader's utility corresponds to its opportunity gain $g(\cdot)$ in validatory power and its bribing costs: 
%\begin{equation}
%    \Pi_0 = g(k) - k \cdot b\enspace , 
%\end{equation}
%where $k$ is the number of exiting validators. The gain can be computed as follows:
%\begin{equation}
%    g(k) = \mathrm{R}(N-k) \cdot \mathrm{PV}(r,Y)
%\end{equation}
%where $R(n)$ refers to the total annual staking reward for a single validator with $n$ active validators, and PV$(r,Y) = \frac{1-(1+r)^{-Y}}{r}$ is the multiplier needed to calculate the net Present Value of $R(n)$ with discount factor $r$ and a $Y$ years horizon~\cite{Brealey2020}. 
%We use $r=0.08$: higher than the return of risk-free treasury bonds ($4-5\%$) given the volatility of Ethereum, but lower than that of the lesser-known ``risky assets'' ($>10\%$), and matching the crypto staking literature~\cite{cong2025tokenomics}. A good ``effective horizon'' $Y$ to calculate with is the PV half-life, which is $\approx9$ years in this case,~\cf\Cref{sec:halflife}. Now, $R$ can be closely estimated by the sum of expected annual attestation and block rewards yielding: 

\begin{equation}\label{eq:leader_utility}
    \Pi_0 = \underbrace{g(k)}_{\text{opp. gain}} + \underbrace{\Pi_0^e}_{\text{extra profit}} - \underbrace{k \cdot b}_{\text{cost}} = \alpha \cdot N  (R(N-k) - R(N)) \cdot \mathrm{PV}(r,Y) + \Pi_0^e - k \cdot b 
\end{equation}
where 
%\begin{equation}
%    \label{eq:util_L}
%    \Pi_0 = \overbrace{\underbrace{\underbrace{32}_{stake} \cdot \bigg(\underbrace{\frac{2940.21}{\sqrt{N-k}}}_{\text{attestation rew.}} + \underbrace{\frac{1078543.3}{N-k}}_{\text{block proposal rew.}}\bigg)}_{\mathrm{R}(N-k)} \cdot \underbrace{\frac{1.08^{-9}}{0.08\cdot100}}_{\mathrm{PV}(r,Y)}}^{\text{opportunity gain }g(k)} - \overbrace{b \cdot k}^{\text{cost}}

\begin{equation}
    \label{eq:util_L}
    R(n) = {\underbrace{32}_{stake} \cdot \bigg(\underbrace{\frac{2940.21}{\sqrt{n}}}_{\text{protocol rewards}} + \underbrace{\frac{1078543.3}{n}}_{\text{estimated MEV}}\bigg)} \quad \mathrm{and} \quad \mathrm{PV}(r,Y) = \frac{1.08^{-9}}{0.08\cdot100}
\end{equation}
\paragraph{Follower's Utility. } A follower's (validator's) utility, $\Pi_i$, is defined in~\Cref{eq:util_F}, depending on their choice to either exit or stay. We define the utility for each action in line with our assumption on a worst-case flat bribe $b$: assuming $k-1$ other validators have already exited. The utility when the follower exits is the value of the bribe $b$: the exiting validator forgoes its future staking rewards, realizing an opportunity cost. %This cost is equivalent to the gain function $c(k)=g(k-1)$, as the validator is the $k$-th to consider exiting. 
Conversely, the utility when staying is the present value of the increased future staking rewards that result from $k-1$ other validators leaving the active set. 

{\renewcommand{\arraystretch}{1.2}
\begin{equation}
    \label{eq:util_F}
    \Pi_i = \left\{\begin{array}{lcl}
        b & \text{ if } & S_i = \mathcal{E} \\
        R(n-k+1) \cdot \mathrm{PV}(r,Y) & \text{ if } & S_i = \mathcal{S} 
    \end{array}\right.
\end{equation}}

%The followers' utility if exiting is the accepted bribe minus the opportunity cost $c(\cdot)$ of giving up future staking rewards:
%\begin{equation}
%    \Pi_i^{\mathcal{E}} = b - c(k) = b - g(k-1)\enspace,
%\end{equation}
%as it corresponds to the $k$-th potentially exiting validator (in line with our assumption on a worst-case flat bribe $b$).
%If staying, the follower will enjoy an opportunity gain $g$ in the form of increased staking rewards:
%\begin{equation}
%    \Pi_i^{\mathcal{S}} = g(k-1)\enspace,
%\end{equation}
%with $k-1$ validators exiting.

\paragraph{Equilibrium analysis. } We solve the game using backward induction. We first find the equilibrium in the follower subgame and then substitute this result into the leader's problem to determine the optimal bribe, $b^*$. An equilibrium in the follower subgame is a state where no validator has an incentive to unilaterally change their strategy, \ie to switch from exit to stay and vice versa. For an equilibrium state with $k$ exiting validators, this means $b - R(n-k+1) \cdot \mathrm{PV}(r,Y) \geq 0$ and $R(n-k) \cdot \mathrm{PV}(r,Y) - b \leq 0$. These conditions establish the range for a bribe $b$ that results in exactly $k$ exits. Hence, to incentivize the precise number of desired exits (\ie $k^*$), the leader must offer an optimal bribe $b^*$ such that

\begin{equation}
\label{eq:eq_bribe}
    R(n-k+1) \cdot \mathrm{PV}(r,Y)\leq b^* \leq R(n-k) \cdot \mathrm{PV}(r,Y).
\end{equation}

Note that a rational briber requires $\Pi_0 > 0$ to participate in the game. A quick calculation shows that this requires extra exogenous profit $\Pi_0^e > 0$ for all optimal bribes $b^*$. In any realistic (and now standard) cryptocurrency adversarial model, one endows the adversary with exogenous motivations~\cite{ford2019rationality}.   

%{\renewcommand{\arraystretch}{1.2}
%\begin{equation}
%    \label{eq:eq_bribe}
%    \forall {i}: S_i^* = \left\{\begin{array}{ll}
%       \mathcal{E} & \phantom{opera}\text{if } b - g(k-1) \ge g(k-1) \\
%       \mathcal{S} & \phantom{opera}\text{if } b - g(k-1) \le g(k-1) \\
%    \end{array}\right.
%\end{equation}}

%We solve the game via backward induction, first solving the follower game, then writing it back to find the optimal leader strategy, \ie bribe $b^*$. The follower game is characterized by an equilibrium state in which no follower has the incentive to deviate unilaterally from its strategy. In such a state, the following best response conditions hold 
%\begin{equation}
%    \forall {i} | S_i^* = \mathcal{E}: b - g(k-1) \geq g(k+1) \quad \mathrm{and} \quad \forall{i} | S_i^* = \mathcal{S}: g(k-1) \geq b - g(k)\enspace, 
%\end{equation}
%which yields the equilibrium bribe for the exact number of exits desired $k^*$:
%\begin{equation}\label{eq:exit_equilibrium_bribe}
%    b^* = g(k^*) + g(k^*+1)\enspace.
%\end{equation}

\begin{figure}[t!]
    \centering
    \subfloat[$\mathsf{PayToExit}$ optimal bribe costs]{{
        \resizebox{0.49\textwidth}{!}{% This file was created with matplot2tikz v0.4.0.
\begin{tikzpicture}

\definecolor{darkgray176}{RGB}{176,176,176}

\begin{axis}[
colorbar,
colorbar style={ylabel={\textbf{Bribe Amount (Billion USD)}}},
colormap={mymap}{[1pt]
  rgb(0pt)=(1,1,0.898039215686275);
  rgb(1pt)=(1,0.968627450980392,0.737254901960784);
  rgb(2pt)=(0.996078431372549,0.890196078431372,0.568627450980392);
  rgb(3pt)=(0.996078431372549,0.768627450980392,0.309803921568627);
  rgb(4pt)=(0.996078431372549,0.6,0.16078431372549);
  rgb(5pt)=(0.925490196078431,0.43921568627451,0.0784313725490196);
  rgb(6pt)=(0.8,0.298039215686275,0.00784313725490196);
  rgb(7pt)=(0.6,0.203921568627451,0.0156862745098039);
  rgb(8pt)=(0.4,0.145098039215686,0.0235294117647059)
},
point meta max=1092.56154311449,
point meta min=0.613841538136861,
tick align=outside,
tick pos=left,
x grid style={darkgray176},
xlabel={\(\displaystyle \boldsymbol{\alpha}\)},
xmin=0, xmax=100,
xtick style={color=black},
xtick={0,9,19,29,39,49,59,69,79,89,99},
xticklabel style={rotate=45.0},
xticklabels={0.01,0.05,0.1,0.15,0.2,0.25,0.3,0.35,0.4,0.45,0.5},
y grid style={darkgray176},
ylabel={\(\displaystyle \boldsymbol{\alpha^*}\)},
ymin=0, ymax=100,
ytick style={color=black},
ytick={0,9,19,29,39,49,59,69,79,89,99},
yticklabels={0.01,0.05,0.1,0.15,0.2,0.25,0.3,0.35,0.4,0.45,0.5}
]
\addplot graphics [includegraphics cmd=\pgfimage,xmin=0, xmax=100, ymin=0, ymax=100] {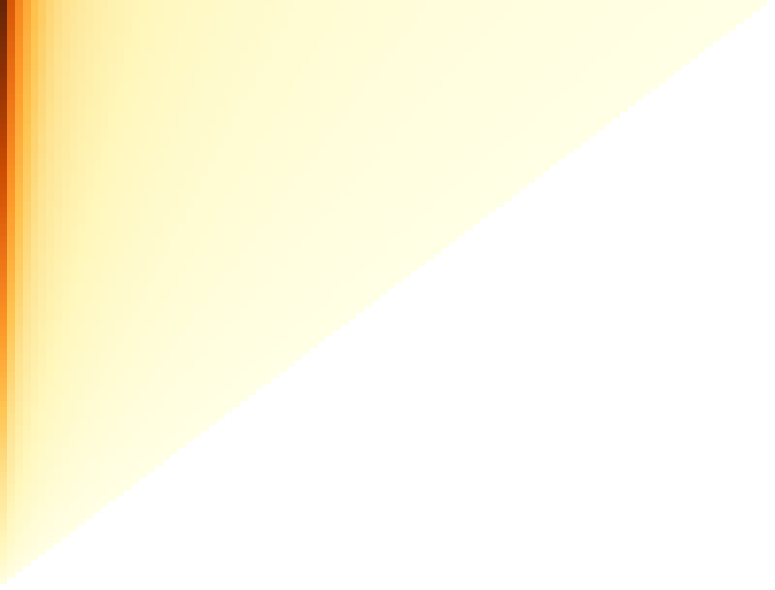};
\draw (axis cs:5.5,14.5) node[
  scale=0.75,
  text=black,
  rotate=0.0
]{\bfseries 59};
\draw (axis cs:5.5,23.5) node[
  scale=0.75,
  text=black,
  rotate=0.0
]{\bfseries 103};
\draw (axis cs:14.5,23.5) node[
  scale=0.75,
  text=black,
  rotate=0.0
]{\bfseries 29};
\draw (axis cs:5.5,32.5) node[
  scale=0.75,
  text=black,
  rotate=0.0
]{\bfseries 141};
\draw (axis cs:14.5,32.5) node[
  scale=0.75,
  text=black,
  rotate=0.0
]{\bfseries 53};
\draw (axis cs:23.5,32.5) node[
  scale=0.75,
  text=black,
  rotate=0.0
]{\bfseries 20};
\draw (axis cs:5.5,41.5) node[
  scale=0.75,
  text=black,
  rotate=0.0
]{\bfseries 176};
\draw (axis cs:14.5,41.5) node[
  scale=0.75,
  text=black,
  rotate=0.0
]{\bfseries 74};
\draw (axis cs:23.5,41.5) node[
  scale=0.75,
  text=black,
  rotate=0.0
]{\bfseries 37};
\draw (axis cs:32.5,41.5) node[
  scale=0.75,
  text=black,
  rotate=0.0
]{\bfseries 15};
\draw (axis cs:5.5,50.5) node[
  scale=0.75,
  text=black,
  rotate=0.0
]{\bfseries 209};
\draw (axis cs:14.5,50.5) node[
  scale=0.75,
  text=black,
  rotate=0.0
]{\bfseries 93};
\draw (axis cs:23.5,50.5) node[
  scale=0.75,
  text=black,
  rotate=0.0
]{\bfseries 51};
\draw (axis cs:32.5,50.5) node[
  scale=0.75,
  text=black,
  rotate=0.0
]{\bfseries 28};
\draw (axis cs:41.5,50.5) node[
  scale=0.75,
  text=black,
  rotate=0.0
]{\bfseries 12};
\draw (axis cs:5.5,59.5) node[
  scale=0.75,
  text=black,
  rotate=0.0
]{\bfseries 241};
\draw (axis cs:14.5,59.5) node[
  scale=0.75,
  text=black,
  rotate=0.0
]{\bfseries 110};
\draw (axis cs:23.5,59.5) node[
  scale=0.75,
  text=black,
  rotate=0.0
]{\bfseries 65};
\draw (axis cs:32.5,59.5) node[
  scale=0.75,
  text=black,
  rotate=0.0
]{\bfseries 40};
\draw (axis cs:41.5,59.5) node[
  scale=0.75,
  text=black,
  rotate=0.0
]{\bfseries 23};
\draw (axis cs:50.5,59.5) node[
  scale=0.75,
  text=black,
  rotate=0.0
]{\bfseries 10};
\draw (axis cs:5.5,68.5) node[
  scale=0.75,
  text=black,
  rotate=0.0
]{\bfseries 272};
\draw (axis cs:14.5,68.5) node[
  scale=0.75,
  text=black,
  rotate=0.0
]{\bfseries 127};
\draw (axis cs:23.5,68.5) node[
  scale=0.75,
  text=black,
  rotate=0.0
]{\bfseries 78};
\draw (axis cs:32.5,68.5) node[
  scale=0.75,
  text=black,
  rotate=0.0
]{\bfseries 51};
\draw (axis cs:41.5,68.5) node[
  scale=0.75,
  text=black,
  rotate=0.0
]{\bfseries 32};
\draw (axis cs:50.5,68.5) node[
  scale=0.75,
  text=black,
  rotate=0.0
]{\bfseries 19};
\draw (axis cs:59.5,68.5) node[
  scale=0.75,
  text=black,
  rotate=0.0
]{\bfseries 9};
\draw (axis cs:5.5,77.5) node[
  scale=0.75,
  text=black,
  rotate=0.0
]{\bfseries 302};
\draw (axis cs:14.5,77.5) node[
  scale=0.75,
  text=black,
  rotate=0.0
]{\bfseries 143};
\draw (axis cs:23.5,77.5) node[
  scale=0.75,
  text=black,
  rotate=0.0
]{\bfseries 90};
\draw (axis cs:32.5,77.5) node[
  scale=0.75,
  text=black,
  rotate=0.0
]{\bfseries 61};
\draw (axis cs:41.5,77.5) node[
  scale=0.75,
  text=black,
  rotate=0.0
]{\bfseries 42};
\draw (axis cs:50.5,77.5) node[
  scale=0.75,
  text=black,
  rotate=0.0
]{\bfseries 27};
\draw (axis cs:59.5,77.5) node[
  scale=0.75,
  text=black,
  rotate=0.0
]{\bfseries 17};
\draw (axis cs:68.5,77.5) node[
  scale=0.75,
  text=black,
  rotate=0.0
]{\bfseries 8};
\draw (axis cs:5.5,86.5) node[
  scale=0.75,
  text=black,
  rotate=0.0
]{\bfseries 332};
\draw (axis cs:14.5,86.5) node[
  scale=0.75,
  text=black,
  rotate=0.0
]{\bfseries 159};
\draw (axis cs:23.5,86.5) node[
  scale=0.75,
  text=black,
  rotate=0.0
]{\bfseries 101};
\draw (axis cs:32.5,86.5) node[
  scale=0.75,
  text=black,
  rotate=0.0
]{\bfseries 70};
\draw (axis cs:41.5,86.5) node[
  scale=0.75,
  text=black,
  rotate=0.0
]{\bfseries 50};
\draw (axis cs:50.5,86.5) node[
  scale=0.75,
  text=black,
  rotate=0.0
]{\bfseries 35};
\draw (axis cs:59.5,86.5) node[
  scale=0.75,
  text=black,
  rotate=0.0
]{\bfseries 24};
\draw (axis cs:68.5,86.5) node[
  scale=0.75,
  text=black,
  rotate=0.0
]{\bfseries 15};
\draw (axis cs:77.5,86.5) node[
  scale=0.75,
  text=black,
  rotate=0.0
]{\bfseries 7};
\draw (axis cs:5.5,95.5) node[
  scale=0.75,
  text=black,
  rotate=0.0
]{\bfseries 361};
\draw (axis cs:14.5,95.5) node[
  scale=0.75,
  text=black,
  rotate=0.0
]{\bfseries 174};
\draw (axis cs:23.5,95.5) node[
  scale=0.75,
  text=black,
  rotate=0.0
]{\bfseries 112};
\draw (axis cs:32.5,95.5) node[
  scale=0.75,
  text=black,
  rotate=0.0
]{\bfseries 80};
\draw (axis cs:41.5,95.5) node[
  scale=0.75,
  text=black,
  rotate=0.0
]{\bfseries 58};
\draw (axis cs:50.5,95.5) node[
  scale=0.75,
  text=black,
  rotate=0.0
]{\bfseries 43};
\draw (axis cs:59.5,95.5) node[
  scale=0.75,
  text=black,
  rotate=0.0
]{\bfseries 31};
\draw (axis cs:68.5,95.5) node[
  scale=0.75,
  text=black,
  rotate=0.0
]{\bfseries 21};
\draw (axis cs:77.5,95.5) node[
  scale=0.75,
  text=black,
  rotate=0.0
]{\bfseries 13};
\draw (axis cs:86.5,95.5) node[
  scale=0.75,
  text=black,
  rotate=0.0
]{\bfseries 6};
\end{axis}

\end{tikzpicture}}
}\label{fig:pay_to_exit_bribe_heatmap}}%
    \centering
    \subfloat[$\mathsf{PayToExit}$ attack duration]{{
        \resizebox{0.48\textwidth}{!}{% This file was created with matplot2tikz v0.4.0.
\begin{tikzpicture}

\definecolor{darkgray176}{RGB}{176,176,176}

\begin{axis}[
colorbar,
colorbar style={ylabel={\textbf{Attack duration (days)}}},
colormap={mymap}{[1pt]
  rgb(0pt)=(1,1,0.850980392156863);
  rgb(1pt)=(0.929411764705882,0.972549019607843,0.694117647058824);
  rgb(2pt)=(0.780392156862745,0.913725490196078,0.705882352941177);
  rgb(3pt)=(0.498039215686275,0.803921568627451,0.733333333333333);
  rgb(4pt)=(0.254901960784314,0.713725490196078,0.768627450980392);
  rgb(5pt)=(0.113725490196078,0.568627450980392,0.752941176470588);
  rgb(6pt)=(0.133333333333333,0.368627450980392,0.658823529411765);
  rgb(7pt)=(0.145098039215686,0.203921568627451,0.580392156862745);
  rgb(8pt)=(0.0313725490196078,0.113725490196078,0.345098039215686)
},
point meta max=611.743888888889,
point meta min=6.17944444444444,
tick align=outside,
tick pos=left,
x grid style={darkgray176},
xlabel={\(\displaystyle \boldsymbol{\alpha}\)},
xmin=0, xmax=100,
xtick style={color=black},
xtick={0,9,19,29,39,49,59,69,79,89,99},
xticklabels={0.01,0.05,0.1,0.15,0.2,0.25,0.3,0.35,0.4,0.45,0.5},
xticklabel style={rotate=45.0},
y grid style={darkgray176},
ylabel={\(\displaystyle \boldsymbol{\alpha^*}\)},
ymin=0, ymax=100,
ytick style={color=black},
ytick={0,9,19,29,39,49,59,69,79,89,99},
yticklabels={0.01,0.05,0.1,0.15,0.2,0.25,0.3,0.35,0.4,0.45,0.5}
]
\addplot graphics [includegraphics cmd=\pgfimage,xmin=0, xmax=100, ymin=0, ymax=100] {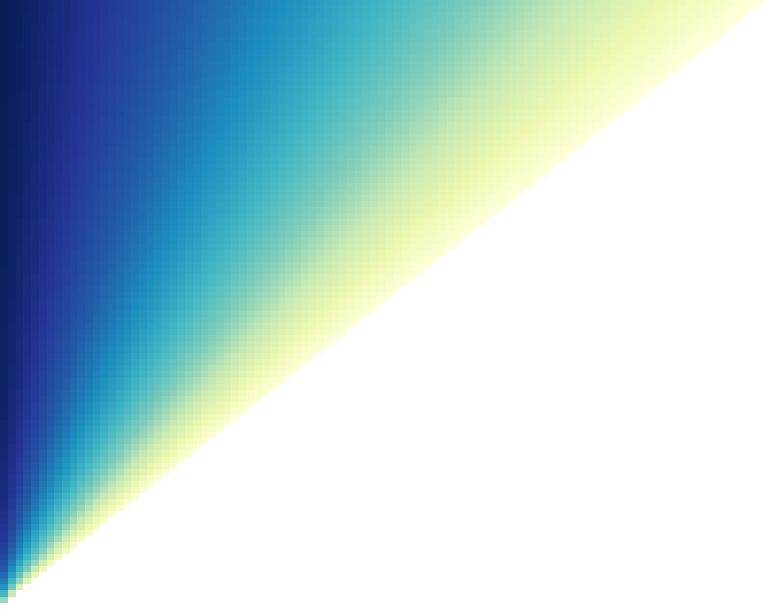};
\draw (axis cs:5.5,15.5) node[
  scale=0.75,
  text=black,
  rotate=0.0
]{\bfseries 367};
\draw (axis cs:5.5,25.5) node[
  scale=0.75,
  text=black,
  rotate=0.0
]{\bfseries 462};
\draw (axis cs:15.5,25.5) node[
  scale=0.75,
  text=black,
  rotate=0.0
]{\bfseries 231};
\draw (axis cs:5.5,35.5) node[
  scale=0.75,
  text=black,
  rotate=0.0
]{\bfseries 506};
\draw (axis cs:15.5,35.5) node[
  scale=0.75,
  text=black,
  rotate=0.0
]{\bfseries 337};
\draw (axis cs:25.5,35.5) node[
  scale=0.75,
  text=black,
  rotate=0.0
]{\bfseries 169};
\draw (axis cs:5.5,45.5) node[
  scale=0.75,
  text=black,
  rotate=0.0
]{\bfseries 531};
\draw (axis cs:15.5,45.5) node[
  scale=0.75,
  text=black,
  rotate=0.0
]{\bfseries 398};
\draw (axis cs:25.5,45.5) node[
  scale=0.75,
  text=black,
  rotate=0.0
]{\bfseries 266};
\draw (axis cs:35.5,45.5) node[
  scale=0.75,
  text=black,
  rotate=0.0
]{\bfseries 133};
\draw (axis cs:5.5,55.5) node[
  scale=0.75,
  text=black,
  rotate=0.0
]{\bfseries 547};
\draw (axis cs:15.5,55.5) node[
  scale=0.75,
  text=black,
  rotate=0.0
]{\bfseries 438};
\draw (axis cs:25.5,55.5) node[
  scale=0.75,
  text=black,
  rotate=0.0
]{\bfseries 328};
\draw (axis cs:35.5,55.5) node[
  scale=0.75,
  text=black,
  rotate=0.0
]{\bfseries 219};
\draw (axis cs:45.5,55.5) node[
  scale=0.75,
  text=black,
  rotate=0.0
]{\bfseries 109};
\draw (axis cs:5.5,65.5) node[
  scale=0.75,
  text=black,
  rotate=0.0
]{\bfseries 559};
\draw (axis cs:15.5,65.5) node[
  scale=0.75,
  text=black,
  rotate=0.0
]{\bfseries 466};
\draw (axis cs:25.5,65.5) node[
  scale=0.75,
  text=black,
  rotate=0.0
]{\bfseries 373};
\draw (axis cs:35.5,65.5) node[
  scale=0.75,
  text=black,
  rotate=0.0
]{\bfseries 279};
\draw (axis cs:45.5,65.5) node[
  scale=0.75,
  text=black,
  rotate=0.0
]{\bfseries 186};
\draw (axis cs:55.5,65.5) node[
  scale=0.75,
  text=black,
  rotate=0.0
]{\bfseries 93};
\draw (axis cs:5.5,75.5) node[
  scale=0.75,
  text=black,
  rotate=0.0
]{\bfseries 567};
\draw (axis cs:15.5,75.5) node[
  scale=0.75,
  text=black,
  rotate=0.0
]{\bfseries 486};
\draw (axis cs:25.5,75.5) node[
  scale=0.75,
  text=black,
  rotate=0.0
]{\bfseries 405};
\draw (axis cs:35.5,75.5) node[
  scale=0.75,
  text=black,
  rotate=0.0
]{\bfseries 324};
\draw (axis cs:45.5,75.5) node[
  scale=0.75,
  text=black,
  rotate=0.0
]{\bfseries 243};
\draw (axis cs:55.5,75.5) node[
  scale=0.75,
  text=black,
  rotate=0.0
]{\bfseries 162};
\draw (axis cs:65.5,75.5) node[
  scale=0.75,
  text=black,
  rotate=0.0
]{\bfseries 81};
\draw (axis cs:5.5,85.5) node[
  scale=0.75,
  text=black,
  rotate=0.0
]{\bfseries 574};
\draw (axis cs:15.5,85.5) node[
  scale=0.75,
  text=black,
  rotate=0.0
]{\bfseries 502};
\draw (axis cs:25.5,85.5) node[
  scale=0.75,
  text=black,
  rotate=0.0
]{\bfseries 430};
\draw (axis cs:35.5,85.5) node[
  scale=0.75,
  text=black,
  rotate=0.0
]{\bfseries 359};
\draw (axis cs:45.5,85.5) node[
  scale=0.75,
  text=black,
  rotate=0.0
]{\bfseries 287};
\draw (axis cs:55.5,85.5) node[
  scale=0.75,
  text=black,
  rotate=0.0
]{\bfseries 215};
\draw (axis cs:65.5,85.5) node[
  scale=0.75,
  text=black,
  rotate=0.0
]{\bfseries 143};
\draw (axis cs:75.5,85.5) node[
  scale=0.75,
  text=black,
  rotate=0.0
]{\bfseries 72};
\draw (axis cs:5.5,95.5) node[
  scale=0.75,
  text=black,
  rotate=0.0
]{\bfseries 579};
\draw (axis cs:15.5,95.5) node[
  scale=0.75,
  text=black,
  rotate=0.0
]{\bfseries 515};
\draw (axis cs:25.5,95.5) node[
  scale=0.75,
  text=black,
  rotate=0.0
]{\bfseries 450};
\draw (axis cs:35.5,95.5) node[
  scale=0.75,
  text=black,
  rotate=0.0
]{\bfseries 386};
\draw (axis cs:45.5,95.5) node[
  scale=0.75,
  text=black,
  rotate=0.0
]{\bfseries 322};
\draw (axis cs:55.5,95.5) node[
  scale=0.75,
  text=black,
  rotate=0.0
]{\bfseries 257};
\draw (axis cs:65.5,95.5) node[
  scale=0.75,
  text=black,
  rotate=0.0
]{\bfseries 193};
\draw (axis cs:75.5,95.5) node[
  scale=0.75,
  text=black,
  rotate=0.0
]{\bfseries 129};
\draw (axis cs:85.5,95.5) node[
  scale=0.75,
  text=black,
  rotate=0.0
]{\bfseries 64};
\end{axis}

\end{tikzpicture}\label{fig:pay_to_exit_attack_duration}}
    }}%
    \caption{$\mathsf{PayToExit}$ bribe costs (in USD) and attack duration (in days) whenever the briber wants to increase its relative staking power from $\alpha$ to $\alpha^{*}$, ($\alpha\leq\alpha^{*}$).}%
    \label{fig:pay_to_exit_measurements}
\end{figure}

\paragraph{Numerical Example. } At the time of writing, there are approximately $N = \num{1123611}$ active validators, the largest staking entity (Lido) controls about $23.9\%$ of the stake. If this entity aimed to amass a $33\%$ share -- a threshold that could threaten the liveness of the Ethereum protocol -- it would need to compel a significant number of other validators to exit. Our calculations show that to reach this target, a total of $k^* \approx \num{317982}$ validators would need to leave the active validator set. Using the equilibrium bribe $b^*$ from~\Cref{eq:eq_bribe}, this scenario would require a bribe of $9.23$ ETH per validator. As of September 20, 2025, this translates to approximately $\num{41332.91}$ USD for each exiting validator. To make this scenario plausible for a rational briber, the extra exogenous profit should be more than $\num{250541.1}$ $\frac{ETH}{\mathsf{year}}$, \revision{thus, limiting the attack's realistic threat profile.}

\subsection{$\mathsf{PayToBias}$ incentives}\label{sec:pay_to_bias_incentives}
\revision{We give a worst-case upper bound on the expected bribe amount whenever a bribee with $(1-\alpha)\beta$ stake auctions off the manipulative power of $k$ consecutive tail slots in epoch $e$. The briber must compensate the bribee with the costs incurred in epoch $e$ and $e+2$, \ie let us denote this by $c_e$ and $c_{e+2}$. First, in epoch $e$, in the worst case, the briber may request the bribee to forfeit proposing blocks in all its $k$ consecutive slots. Thus, the briber must pay at least the block rewards, transaction fees, missed attestation rewards, and maximal extractable value that the bribee gives up in those $k$ tail slots. We estimate $c_e\approx k\cdot \mathsf{v_b}$, where $\mathsf{v_b}$ is the value of a single Ethereum block and the attestation rewards of the bribee's validators. Second, the bribee could have proposed $X\sim Binom(32, \alpha)$ blocks in epoch $e+2$, but with the different seed $Y\sim Binom(32,\alpha)$. In the worst case $X>Y$, we shall pay for the difference $X-Y$, which is expected to be $c_2\approx\mathbb{E}(\max(X-Y,0))\cdot\mathsf{v_b}$. Concretely, for $1-\alpha=0.2,k=3$, we have that the expected bribe amount is $c_e+c_{e+2}\approx 0.965~\mathit{ETH}$.}
\revision{A thorough analysis of} the economic incentives in $\mathsf{PayToBias}$ markets, often called RANDAO bribery markets, are the subject of extensive study in concurrent research, such as the work of Alptürer~\cite{alpturer2025RANDAO}. Given this focus in other contemporary work, a detailed (game-theoretic) analysis of this specific market is beyond the scope of our paper.

%$\mathsf{PayToBias}$ markets (often referred to as the RANDAO bribery markets) are extensively studied in concurrent works,~\eg by Alptürer~\cite{alpturer2025RANDAO}. Therefore, in this work, we refrain from analyzing the complex incentives implied by those markets.

\subsection{Countermeasures}\label{sec:countermeasures}

We envision several potential countermeasures, which can be categorized as either general economic deterrents or targeted protocol upgrades.

\paragraph{General Defenses. } These defenses apply broadly to all three of our proposed attacks by altering the economic incentives of rational validators, thereby protecting safety, liveness, and fairness.

\begin{description}
    \item[Increased Protocol Rewards] A natural countermeasure is to increase the rewards for honest participation. This raises the opportunity cost of deviation, making any bribery attack more expensive. However, this approach likely requires higher inflation, and the trade-off between protocol security and economic stability warrants further study.
    \item[Whistleblowing Incentives] Bribery attacks, as a form of collusion, can be thwarted by enabling participants to whistleblow in a publicly verifiable manner. Protocol-level incentives for whistleblowing could disrupt any bribery market, as explored in~\cite{kelkar2025breaking}.
\end{description}

\paragraph{Targeted Defenses. }These defenses are specific consensus upgrades that would mitigate individual attacks by making them technically infeasible.

\begin{description}
    \item[\revision{Faster} Finality] 
    \revision{To enhance safety, new consensus protocols aim to achieve faster finality, \eg three slots~\cite{d20243}, making long-range attacks impossible. However, not even these can prevent (ex-post) reorgs on the shorter, unfinalized head of the blockchain enabled by $\mathsf{PayToAttest}$.}
    %\revision{We note that simultaneous deterministic finality and dynamic availability in an asynchronous permissionless network is indeed impossible. However, to reduce the attack surface, hybrid protocols such as Orbit SSF would mitigate a large portion of bribery attempts enabled by the PayToAttest attack.}
    %To protect safety, a proposal for single-slot finality~\cite{d2023simple} would finalize blocks almost immediately after they are proposed. This would prevent the reorganizations enabled by the $\mathsf{PayToAttest}$ attack.
    \item[Unbiasable Randomness] To protect fairness, the deployment of a cryptographically secure and unbiasable randomness beacon would directly mitigate the $\mathsf{PayToBias}$ attack by removing the beacon manipulator's ability to influence proposer selection.
    \item[Time-locked Stake Retrieval] To protect liveness, the protocol could enforce a significantly extended withdrawal period for an exiting validator's stake. The long delay imposes a direct economic cost due to the time value of money, increasing the opportunity cost of exiting and making the $\mathsf{PayToExit}$ attack financially impractical.
\end{description}

%We detail possible countermeasures against bribery attacks as follows.
%\begin{description}
%    \item[Increased inflation and protocol rewards] A natural countermeasure for the protocol designer to defend against short-term consensus manipulation attacks is to increase the protocol rewards for all validators. However, increasing inflation possibly causes downward pressure on the price of ether, which is undesirable. We leave it to future work to assess the level of necessary inflation and to quantify the trade-off between inflation and the effectiveness of increased protocol robustness against trustless bribery attacks.
%    \item[Motivating whistleblowing] Bribery attacks, sometimes collusions, can be thwarted by enabling any bribery participant to whistleblow in a publicly verifiable manner.  A potential mitigation could be to incentivize whistleblowing on the protocol level as it was extensively studied recently in~\cite{cryptoeprint:2025/1582}. 
%    \item[Consensus upgrades] A recently discussed protocol change in the community is the single-slot finality proposal~\cite{d2023simple}. If a block was final a single slot later than it had been proposed, then forks and reorgs would not be possible anymore. Thus, it would make our proposed $\mathsf{PayToAttest}$ contract likely obsolete. Similarly, our $\mathsf{PayToBias}$ bribery contract would not be possible, if the Ethereum protocol had deployed a cryptographically secure and unbiasable distributed randomness beacon to select block proposer validators.
%\end{description}
\revision{
\section{Related work}\label{sec:related_work}
The first works studying short-term consensus manipulations (\cf~\Cref{tab:attacks_qualitative_comparison}), \eg forking, censorship, etc., enabled by bribing focused on Bitcoin~\cite{DBLP:conf/aft/AvarikiotiKLM24,bonneau2016buy,bonneau2018hostile,hu2023novel}, \ie Nakamoto consensus. Expressing the bribery contract logic on Bitcoin is cumbersome, thus, follow-up works have moved the bribery contract to a different, Turing-complete chain, \eg Ethereum, largely limiting these bribery attacks' practicality~\cite{judmayer2019pay,judmayer2021sok}. We are interested in bribery contracts and the implied incentives, where both the briber and the bribee act on the same chain~\cite{karakostas2024blockchain}. McCorry et al. designed bribery contracts on Ethereum PoW for the first time~\cite{mccorry2018smart}. Since then, Ethereum has transitioned to a PoS consensus. Thus, we continue their seminal work, but in the setting of Ethereum PoS. Many works analyze the negative externalities of the technically easiest form of bribery: censorship~\cite{DBLP:conf/sigecom/BergerFMS25,wahrstatter2024blockchain,wang2023blockchain}. Closest to our work is that of~\cite{DBLP:conf/eurosp/SarencheTMSP25}, in which bribery attacks on Ethereum PoS had been suggested; here, the bribee is only provided with game-theoretic guarantees (not cryptographic) that the briber will compensate her. To overcome these limitations, we design, implement, and evaluate novel, \emph{trustless} bribery contracts that are not verifiable in PoW ($\mathsf{PayToExit}$), and were not efficient prior to recent protocol upgrades ($\mathsf{PayToAttest}$),~\cf~\Cref{sec:pos_ethereum_prelims}.
}
\section{Conclusion and Future Directions}\label{sec:conclusion}

In this work, we demonstrated that the features enabling scalable, permissionless blockchains -- namely expressive smart contracts and efficient consensus mechanisms -- can be repurposed to create novel attack vectors. We designed, implemented, and evaluated three trustless bribery contracts ($\mathsf{PayToAttest}$, $\mathsf{PayToExit}$, and $\mathsf{PayToBias}$) that allow an adversary to undermine core consensus properties in Ethereum. Our initial game-theoretic analysis of the $\mathsf{PayToExit}$ market quantifies the practical incentives for validators to participate in such attacks.

Our findings open several promising avenues for future research:
\begin{description}
    \item[Anonymity and Privacy] The details of our bribery contracts (\eg participants, validator indices, and bribe amounts) are public. These could be concealed using confidential transaction schemes~\cite{bunz2020zether} and zero-knowledge proofs to increase the stealth of such attacks.
    \item[Additional Bribery Contracts] The interface we developed,~\cf~\Cref{fig:bribe_interface}, can be used to implement other known bribery attacks, such as a $\mathsf{PayToInclude}$ contract that pays proposers to include or censor specific transactions~\cite{DBLP:conf/sigecom/BergerFMS25,wahrstatter2024blockchain,wang2023blockchain}.
    \item[Game-Theoretic Extensions] Our analysis of the $\mathsf{PayToExit}$ market serves only as a starting point. Promising extensions include modeling temporal dynamics, variable bribes, and heterogeneous validators. Furthermore, formal game models for the $\mathsf{PayToAttest}$ and $\mathsf{PayToBias}$ attacks are needed to fully characterize the related economic incentives.
    \item[Attack Robustness vs. Inflation] While increasing protocol rewards through inflation could potentially deter bribery, the relationship between a cryptocurrency's monetary policy and its resilience to such economic attacks is not yet well understood and warrants further study.
\end{description}

\ifanonymous
\else
\paragraph{Acknowledgements.}
We are grateful to Kaya Alptürer for insightful discussions and encouragement. We are thankful to the Financial Cryptography and Data Security 2026 reviewers for their constructive feedback that helped improve this manuscript. István András Seres was supported by the Ministry of Culture and Innovation and the National Research, Development, and Innovation Office within the Quantum Information National Laboratory of Hungary (Grant No. 2022-2.1.1-NL-2022-00004).
Balazs Pejo was supported by the European Union project RRF-2.3.1-21-2022-00004 within the framework of the Artificial Intelligence National Laboratory.
Gergely Biczók was supported by Project no. 138903 of the Ministry of Innovation and Technology, Hungary, from the NRDI
Fund, financed under the FK\_21 funding scheme.
\fi

\bibliography{sample}
\bibliographystyle{plain}

\appendix
\section{Additional Preliminaries}\label{sec:additional_preliminaries}

\subsection{Additional Notations}\label{sec:additional_notations}

The function $\mathbf{blockhash}(n):\mathbb{N}\rightarrow\{0,1\}^{256}$ takes a block height $n$ as input and returns the corresponding block hash.

Ethereum extensively uses Merkle trees~\cite{merkle1987digital} as vector commitments or cryptographic accumulators. For instance, validator public keys are committed to in an incremental Merkle tree whose root hash is accessible on-chain to any smart contract. We denote the relevant Merkle tree functionalities as follows.

\begin{description}
    \item[$\mathsf{Merkle.Insert}(T,x)\rightarrow T'$.] Adds element $x$ as a new leaf node to the Merkle tree $T$ and returns an updated Merkle tree $T'$.
    \item[$\mathsf{Merkle.Prove}(T,x,i)\rightarrow\pi$.] Generates a proof $\pi$ proving the membership of element $x$ at the index $i$ in Merkle tree $T$.
    \item[$\mathsf{Merkle.Verify}(\mathsf{root},x,i,\pi)\rightarrow\{0,1\}$.] Given a root hash $\mathsf{root}$ of the Merkle tree, this algorithm verifies the proof $\pi$ for the statement that $x$ is committed in the $i$th position in the Merkle tree.
\end{description}

\subsection{On-chain BLS batch verification}\label{sec:bls_batch_on_chain}

Our $\mathsf{PayToAttest}$ contract requires the on-chain batch verification of BLS signatures from validator attestations. The practical implementation of this contract is enabled by EIP-7549, the importance of which we now empirically demonstrate. EIP-7549 ensures that when validators attest to the same block header, they sign the exact same message. This allows for constant-time ($\mathcal{O}(1)$) BLS batch verification, independent of the number of signers, as shown in~\Cref{fig:bls_batch_verification_measurements} and~\Cref{eq:bls_same_message_check}. In contrast, before this EIP, the linear-cost verification would have made our $\mathsf{PayToAttest}$ contract prohibitively expensive for bribing a large number of validators.

\begin{figure}[ht]
    \centering
    \subfloat[BLS batch verification gas costs. At the time of writing, the Ethereum block gas limit is set to $\num{45087931}$ gas.]{{
        \resizebox{0.45\textwidth}{!}{\input{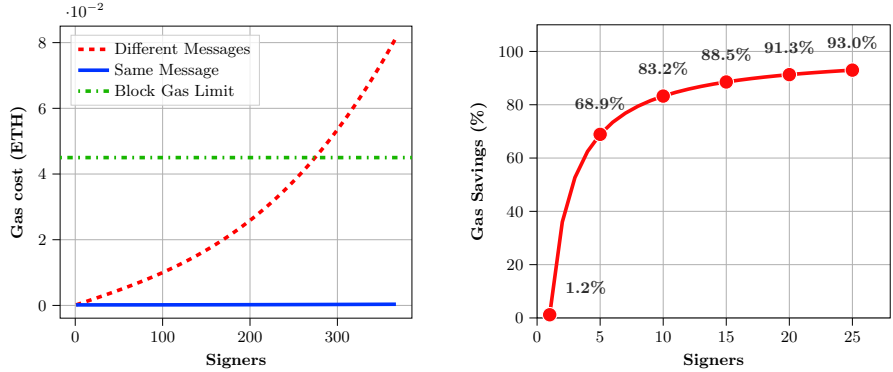}}
    }}%
    \hspace{3mm}
    \subfloat[BLS batch verification gas savings when the same message $m$ vs. different $\{m_i\}_i$ messages are signed.]{{
        \resizebox{0.47\textwidth}{!}{% This file was created with matplot2tikz v0.4.0.
\begin{tikzpicture}

\definecolor{darkgray176}{RGB}{176,176,176}
\definecolor{darkslategray46}{RGB}{46,46,46}
\definecolor{red2551010}{RGB}{255,10,10}

\begin{axis}[
tick align=outside,
tick pos=left,
x grid style={darkgray176},
xlabel={$ \displaystyle \textbf{Signers}$},
xmajorgrids,
xmin=0, xmax=28,
xtick style={color=black},
y grid style={darkgray176},
ylabel={$ \displaystyle \textbf{Gas Savings (\%)}$},
ymajorgrids,
ymin=0, ymax=110,
ytick style={color=black}
]
\addplot [line width=2pt, red2551010]
table {%
1 1.19904515955726
2 36.0043022317827
3 52.6796365167357
4 62.4618495876509
5 68.8932232866051
6 73.4439588992559
7 76.8328592604463
8 79.4562117506166
9 81.5455925181509
10 83.2499224485575
11 84.6653549455011
12 85.8615996649916
13 86.8845389235414
14 87.7699513163592
15 88.543211969652
16 89.2248342118687
17 89.8295333475389
18 90.3700777312719
19 90.8567699697258
20 91.2971733493346
21 91.6958173272957
22 92.0605930382467
23 92.3953023165774
24 92.701751364385
25 92.9859525143276
};
\addplot [semithick, red2551010, mark=*, mark size=4, mark options={solid,draw=white}, only marks]
table {%
1 1.19904515955726
};
\addplot [semithick, red2551010, mark=*, mark size=4, mark options={solid,draw=white}, only marks]
table {%
5 68.8932232866051
};
\addplot [semithick, red2551010, mark=*, mark size=4, mark options={solid,draw=white}, only marks]
table {%
10 83.2499224485575
};
\addplot [semithick, red2551010, mark=*, mark size=4, mark options={solid,draw=white}, only marks]
table {%
15 88.543211969652
};
\addplot [semithick, red2551010, mark=*, mark size=4, mark options={solid,draw=white}, only marks]
table {%
20 91.2971733493346
};
\addplot [semithick, red2551010, mark=*, mark size=4, mark options={solid,draw=white}, only marks]
table {%
25 92.9859525143276
};
\draw (axis cs:1,1.19904515955726) ++(20pt,12pt) node[
  scale=1,
  anchor=base,
  text=darkslategray46,
  rotate=0.0
]{\bfseries 1.2\%};
\draw (axis cs:5,68.8932232866051) ++(0pt,14pt) node[
  scale=1,
  anchor=base,
  text=darkslategray46,
  rotate=0.0
]{\bfseries 68.9\%};
\draw (axis cs:10,83.2499224485575) ++(0pt,12pt) node[
  scale=1,
  anchor=base,
  text=darkslategray46,
  rotate=0.0
]{\bfseries 83.2\%};
\draw (axis cs:15,88.543211969652) ++(0pt,12pt) node[
  scale=1,
  anchor=base,
  text=darkslategray46,
  rotate=0.0
]{\bfseries 88.5\%};
\draw (axis cs:20,91.2971733493346) ++(0pt,12pt) node[
  scale=1,
  anchor=base,
  text=darkslategray46,
  rotate=0.0
]{\bfseries 91.3\%};
\draw (axis cs:25,92.9859525143276) ++(0pt,12pt) node[
  scale=1,
  anchor=base,
  text=darkslategray46,
  rotate=0.0
]{\bfseries 93.0\%};
\end{axis}
\end{tikzpicture}}
    }}%
    \caption{BLS batch verification measurements on Ethereum after EIP-2537~\cite{vlasov2020eip}. On-chain  BLS batch verification gas cost (left) when the same message (blue line) or different messages (red line) are signed by multiple signers. If different messages are signed, then only fewer than $275$ messages can be batch-verified in a single Ethereum block. Gas savings (right) are already significant for a handful of signed messages (if the same message is signed).}%
    \label{fig:bls_batch_verification_measurements}%
\end{figure}

\subsection{Attestation processing and rewards}\label{sec:attestation_rewards}

During each epoch, a validator is expected to broadcast an attestation to its assigned subnet. As detailed in~\Cref{fig:class_structures}, an attestation represents a validator's view of the chain and contains several components. The $\mathsf{beacon\_block\_root}$ is the validator's vote for the head of the chain under the LMD GHOST fork-choice rule. Concurrently, the $\mathsf{target}$ and $\mathsf{source}$ fields are votes for checkpoints under the Casper FFG finality gadget, which is responsible for finalizing the blockchain.

\begin{table*}[ht]
\centering
\begin{minipage}{\textwidth}
\begin{center}
 \scalebox{1}{
 \begin{tabular}{l c c c c} 
 \toprule
   \textbf{Timeliness}&\textbf{1 slot} & \textbf{$\leq$ 5 slots} & \textbf{$\leq$ 32 slots} & \thead{\textbf{$>$ 32 slots} \\ \textbf{(missing)}} \\ [0.5ex] 
\midrule
\makecell[l]{Wrong source} & 0 & 0 & 0 & 0 \\ Correct source & $W_s$ & $W_s$ & 0 & 0 \\ 
\makecell[l]{Correct source,\\ and target} & $W_s + W_t$ & $W_s + W_t$ & $W_t$ & 0 \\ 
\makecell[l]{Correct source, \\ target and head} & \makecell{$W_s + $ \\ $W_t + W_h$} & $W_s + W_t$ & $W_t$ & 0 \\[1ex] 
\bottomrule
\end{tabular}
}
\caption{Attestation reward matrix for validators, weights, and proportion of a validator's reward for timely votes are as follows: source ($W_s=21.9\% $), target ($W_t=40.6\%$), and head votes ($W_h=21.9\%$), respectively. The timeliness columns indicate when the attestation was included in the canonical chain after the validator needed to broadcast its attestation.}
\label{tab:attestation_rewards}
\end{center}
\end{minipage}
\end{table*}

As~\Cref{tab:attestation_rewards} shows, a validator can earn a partial reward even if their attestation does not perfectly match the canonical chain. The primary condition for receiving any reward is that the attestation itself must be included in a canonical block~\cite{DBLP:conf/eurosp/SarencheTMSP25}. During a fork, a validator risks forfeiting rewards if their attestation supports the losing (i.e., non-canonical) chain. However, during a short fork (fewer than $5$ slots), even if a validator's head vote is incorrect, the validators can still earn a significant partial reward for correct source and target votes. This results in a reward with weight $W_s+W_t$, corresponding to $74\%$ of the maximum possible reward. %\abel{pls reviewzd a változtatásomat. Azért $74\%$, mert attestationnel amúgy sem kapnál többet mint $W_s+W_t+W_h<1$  és $74\%=\frac{W_s+W_t}{W_s+W_t+W_h}$}\istvan{LFG}
\section{Bribery Contract Implementations}\label{sec:contract_implementation}

This section presents the pseudocode for our three proposed bribery contracts. The code focuses on the core logic and omits certain low-level, EVM-specific implementation details for the sake of clarity. For the complete Solidity smart contracts and the accompanying test suite, readers are referred to our publicly available open-source repository.~\footnote{\ifanonymous
~\url{https://anonymous.4open.science/r/bribery-zoo-70FD/}.
\else
~\url{https://github.com/0xSooki/bribery-zoo}.
\fi}

\subsection{The $\mathsf{PayToAttest}$ Bribery Contract}\label{sec:paytoattest_implementation}

In a $\mathsf{PayToAttest}$ contract, a briber purchases attestations for one or more of their proposed blocks. The process begins when the briber calls the \texttt{offerBribe}$(\pk^{*}, m, t)$ function, specifying the aggregate public key $\pk^{*}$ of the target validators, the block header $m$, and a submission deadline $t$. As part of this initial transaction, the briber also transfers the total $\mathsf{bribeAmnt}$ in Ether to the contract. The set of validators corresponding to $\pk^{*}$ can then claim their reward by providing the contract with a valid, aggregate BLS signature on $m$ before the deadline expires, as detailed in \Cref{alg:paytoattest}.

\begin{algorithm}[H]
\caption{The $\mathsf{PayToAttest}$ contract: The pseudocode of the $\texttt{offerBribe}()$ and $\texttt{takeBribe}()$ functions as called by the briber and bribee, respectively.}
\label{alg:paytoattest}
\begin{algorithmic}[1]
\STATE \textbf{function} \texttt{offerBribe}($\pk^{*} \in \mathbb{G}_1,m \in \mathsf{AttestationData},t \in \mathbb{N}$):
  \STATE\hspace{\algorithmicindent} $\pk_{agg}=\pk^{*}$ //Aggregate public key of the bribed validators $\{\pk_i\}^{n}_{i=1}$.
  \STATE\hspace{\algorithmicindent} $\mathsf{block.header}=m$ // This is the $\mathsf{AttestationData}$,~\cf~\Cref{fig:class_structures}.
  \STATE\hspace{\algorithmicindent} $\mathsf{deadline}=t$ //The briber only rewards timely votes.
  \STATE\hspace{\algorithmicindent} Transfer $\mathsf{bribeAmnt}$ to the contract.
\STATE
\STATE \textbf{function} \texttt{takeBribe}($\sigma \in \mathbb{G}_2$):
\begin{ALC@g}
    \STATE \textbf{assert} $!\mathsf{claimed}[\sigma]$
    \STATE \textbf{assert} $\mathsf{block.timestamp} < \mathsf{deadline}$.
    \STATE \textbf{assert} $\bls.\verifyy(\pk_{agg},m,\sigma,)$ //This is implemented applying the same-message BLS batch verification check,~\cf~\Cref{eq:bls_same_message_check}.
    \STATE $\mathsf{claimed}[\sigma]=\mathsf{true}$
\STATE Transfer $\mathsf{bribeAmnt}$ ether to validator.
\end{ALC@g}
\end{algorithmic}
\end{algorithm}

The $\mathsf{PayToAttest}$ contract can be extended for greater flexibility by removing the block header $m$ as an argument from the \texttt{offerBribe()} function, resulting in the signature \texttt{offerBribe}$(\pk^{*}, t)$. This change allows the briber to decide on the target block header after the bribe has been offered, but it necessitates an additional verification step within the \texttt{takeBribe}($\sigma$) function. To prevent misuse, the contract must confirm that the attestation provided by the bribee is for a block actually proposed by the briber. This can be achieved by checking the block's randomness field -- which contains the proposer's BLS signature on the epoch number (\cf~\Cref{sec:RANDAO_bribery_market}) -- to ensure it was signed with the briber's public key.

\subsection{The $\mathsf{PayToExit}$ Bribery Contract}\label{sec:paytoexit_implementation}

The $\mathsf{PayToExit}$ contract creates a trustless bribery market to incentivize Ethereum validators to voluntarily exit the active validator set (\cf~\Cref{sec:pay_to_exit}). This contract rewards a validator with $\mathsf{bribeAmnt}$ in Ether if they can prove that the following three conditions have been met for their validator index $i \in \mathbb{N}$:

\begin{description}
    \item[Membership in the current validator set] The validator must prove they are a current member of the active validator set. This is achieved by providing a Merkle proof $\pi$ showing that their public key, $\pk_i$, is a leaf in the deposit contract's Merkle tree. The contract verifies this proof against the publicly accessible deposit root hash.
    \item[A signed exit transaction] The validator must provide a valid BLS signature on a voluntary exit message (\cf~\Cref{fig:class_structures}). The $\mathsf{PayToExit}$ contract verifies this signature to confirm the validator's authentic intent to exit the protocol.
    \item[Non-membership in the rewarded validator set] To prevent duplicate payments, the contract must verify that the validator with index $i$ has not already been rewarded for exiting. This is implemented by checking for non-membership in the $\mathsf{claimed}[i]$ mapping, which tracks all validators who have already received a payout.
\end{description}

The logic for verifying these three conditions is implemented in the $\mathsf{takeBribe}()$ function, as detailed in \Cref{alg:paytoexit}. If the conditions are satisfied, the contract rewards the validator by distributing the funds accordingly. 

The contract's logic could be extended with additional constraints. For instance, a briber might add a time-based check to only reward voluntary exits that occurred within a specific timeframe. While our current implementation does not include this feature, it represents a straightforward extension.

\begin{algorithm}[H]
\caption{The $\mathsf{PayToExit}$ contract: The pseudocode of the $\texttt{offerBribe}()$ and $\texttt{takeBribe}()$ functions as called by the briber and bribee, respectively.}
\label{alg:paytoexit}
\begin{algorithmic}[1]
\STATE \textbf{function} offerBribe()
  \STATE\hspace{\algorithmicindent} Transfer $\mathsf{bribeCost}$ to the contract.
\STATE
\STATE \textbf{function} takeBribe($i \in \mathbb{N}$, $\sigma \in \mathbb{G}_2$, $\pi \in \mathsf{bytes[]}$)
\begin{ALC@g}
\STATE \textbf{assert} $\mathsf{Merkle.Verify}(\mathsf{root},\pk_i, \ i, \ \pi)$ //$\mathsf{root}$ is read from the deposit contract.
\STATE \textbf{assert} $\bls.\verifyy(\pk_i, m_{\mathsf{exit,i}},\sigma)$ // The message $m_{\mathsf{exit,i}}$ is the \texttt{VoluntaryExit} object for validator $i$,~\cf~\Cref{fig:class_structures}.
\STATE \textbf{assert} $!\mathsf{claimed}[i]$ // $i$ is the validator's deposit index in the deposit contract.
\STATE $\mathsf{claimed}[i]=\mathsf{true}$.
\STATE Transfer $\mathsf{bribeAmnt}$ ether to validator $i$.
\end{ALC@g}
\end{algorithmic}
\end{algorithm}

While our proof-of-concept implementation of the $\mathsf{PayToExit}$ contract uses a fixed $\mathsf{bribeAmnt}$, the design can be easily extended to support a dynamic bribe amount. A straightforward extension would be to make the $\mathsf{bribeAmnt}$ a function of $N(t)$ -- the total number of active validators -- allowing the reward to scale with network participation.

\subsection{The $\mathsf{PayToBias}$ Bribery Contract}\label{sec:paytobias_implementation}

Ethereum uses the distributed randomness beacon RANDAO to select block proposers. The beacon's randomness is known to be biasable by rational validators who can strategically reveal or withhold their randomness contribution~\cite{alpturer2024optimal,nagy2025forking,wahrstatter2023time}. This creates an economic incentive for manipulation, as being selected more often leads to greater rewards.

Our $\mathsf{PayToBias}$ contract creates an efficient market to auction off this manipulation right. For simplicity, our model considers an adversary who controls the final slot of an epoch, effectively auctioning a single bit of bias. Other rational validators can then bid on whether the adversary should publish or withhold their block in that slot. Once the auction concludes, the contract pays the adversary the total amount bid on the winning outcome.

To collect this reward, the adversary must prove to the $\mathsf{PayToBias}$ contract which action they took. This proof must satisfy the following three conditions:

\begin{description}
    \item[Block integrity] The hashes of the provided blocks are verified against the built-in \texttt{blockhash} function in Solidity, ensuring that each block corresponds to its designated block height and that the block contents are authentic.    
    \item[Block linkage] The given two headers must be correspond to consecutive blocks, thus the second block must reference the hash of the first block.
    \item[Time difference] To confirm a block was withheld, the contract must verify that no block was published in the designated final slot of the epoch. This is accomplished by analyzing the timestamps of the surrounding blocks, which are enforced by Ethereum's consensus rules: the contract checks if the timestamp of the first block in the next epoch is more than 12 seconds greater than the timestamp of the penultimate block of the preceding epoch. Since each slot is 12 seconds long, this gap proves that a slot was skipped, confirming the adversary withheld their block as promised.
\end{description}

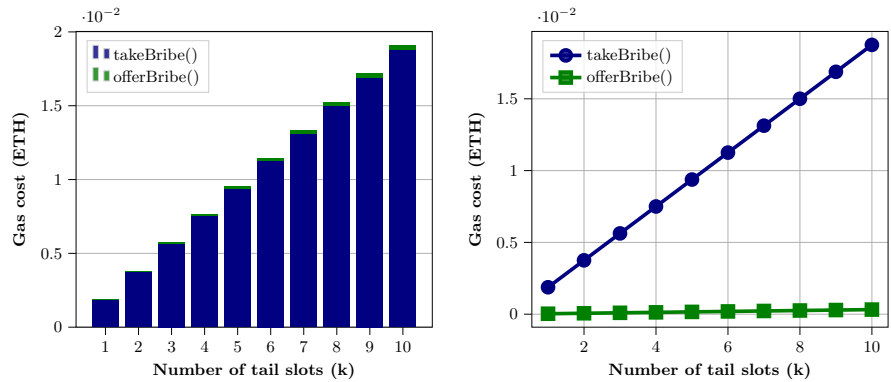
\begin{figure}[ht]
    \centering
    \subfloat[A stacked chart illustrating the linear gas cost in the number of auctioned tail slots in the $\mathsf{takeBribe}(\cdot)$ function.]{{
        \resizebox{0.47\textwidth}{!}{% This file was created with matplot2tikz v0.4.0.
\begin{tikzpicture}

\definecolor{darkgray176}{RGB}{176,176,176}
\definecolor{green}{RGB}{0,128,0}
\definecolor{lightgray204}{RGB}{204,204,204}
\definecolor{navy}{RGB}{0,0,128}

\begin{axis}[
legend cell align={left},
legend style={
  fill opacity=0.8,
  draw opacity=1,
  text opacity=1,
  at={(0.03,0.97)},
  anchor=north west,
  draw=lightgray204
},
tick align=outside,
tick pos=left,
x grid style={darkgray176},
xlabel={\(\displaystyle \textbf{Number of tail slots (k)} \)},
xmin=-0.89, xmax=9.89,
xtick style={color=black},
xtick={0,1,2,3,4,5,6,7,8,9},
xticklabels={1,2,3,4,5,6,7,8,9,10},
y grid style={darkgray176},
ylabel={\(\displaystyle \textbf{Gas cost (ETH)} \)},
ymajorgrids,
ymin=0, ymax=0.02003033487,
ytick style={color=black}
]
\draw[draw=none,fill=navy] (axis cs:-0.4,0) rectangle (axis cs:0.4,0.0018755758);
\addlegendimage{ybar,ybar legend,draw=none,fill=navy}
\addlegendentry{takeBribe()}

\draw[draw=none,fill=navy] (axis cs:0.6,0) rectangle (axis cs:1.4,0.0037511516);
\draw[draw=none,fill=navy] (axis cs:1.6,0) rectangle (axis cs:2.4,0.0056267274);
\draw[draw=none,fill=navy] (axis cs:2.6,0) rectangle (axis cs:3.4,0.0075023032);
\draw[draw=none,fill=navy] (axis cs:3.6,0) rectangle (axis cs:4.4,0.009377879);
\draw[draw=none,fill=navy] (axis cs:4.6,0) rectangle (axis cs:5.4,0.0112534548);
\draw[draw=none,fill=navy] (axis cs:5.6,0) rectangle (axis cs:6.4,0.0131290306);
\draw[draw=none,fill=navy] (axis cs:6.6,0) rectangle (axis cs:7.4,0.0150046064);
\draw[draw=none,fill=navy] (axis cs:7.6,0) rectangle (axis cs:8.4,0.0168801822);
\draw[draw=none,fill=navy] (axis cs:8.6,0) rectangle (axis cs:9.4,0.018755758);
\draw[draw=none,fill=green] (axis cs:-0.4,0.0018755758) rectangle (axis cs:0.4,0.00190765094);
\addlegendimage{ybar,ybar legend,draw=none,fill=green}
\addlegendentry{offerBribe()}

\draw[draw=none,fill=green] (axis cs:0.6,0.0037511516) rectangle (axis cs:1.4,0.00381530188);
\draw[draw=none,fill=green] (axis cs:1.6,0.0056267274) rectangle (axis cs:2.4,0.00572295282);
\draw[draw=none,fill=green] (axis cs:2.6,0.0075023032) rectangle (axis cs:3.4,0.00763060376);
\draw[draw=none,fill=green] (axis cs:3.6,0.009377879) rectangle (axis cs:4.4,0.0095382547);
\draw[draw=none,fill=green] (axis cs:4.6,0.0112534548) rectangle (axis cs:5.4,0.01144590564);
\draw[draw=none,fill=green] (axis cs:5.6,0.0131290306) rectangle (axis cs:6.4,0.01335355658);
\draw[draw=none,fill=green] (axis cs:6.6,0.0150046064) rectangle (axis cs:7.4,0.01526120752);
\draw[draw=none,fill=green] (axis cs:7.6,0.0168801822) rectangle (axis cs:8.4,0.01716885846);
\draw[draw=none,fill=green] (axis cs:8.6,0.018755758) rectangle (axis cs:9.4,0.0190765094);
\end{axis}

\end{tikzpicture}}
}\label{}}% 
\
    \subfloat[Gas measurements for the $\mathsf{PayToBias}$ contracts $\mathsf{offerBribe}()$
and $\mathsf{takeBribe}()$ functions respectively.]{{
        \resizebox{0.47\textwidth}{!}{% This file was created with matplot2tikz v0.4.0.
\begin{tikzpicture}

\definecolor{darkgray176}{RGB}{176,176,176}
\definecolor{green}{RGB}{0,128,0}
\definecolor{lightgray204}{RGB}{204,204,204}
\definecolor{navy}{RGB}{0,0,128}

\begin{axis}[
legend cell align={left},
legend style={
  fill opacity=0.8,
  draw opacity=1,
  text opacity=1,
  at={(0.03,0.97)},
  anchor=north west,
  draw=lightgray204
},
tick align=outside,
tick pos=left,
x grid style={darkgray176},
xlabel={\(\displaystyle \textbf{Number of tail slots (k)} \)},
xmajorgrids,
xmin=0.55, xmax=10.45,
xtick style={color=black},
y grid style={darkgray176},
ylabel={\(\displaystyle \textbf{Gas cost (ETH)} \)},
ymajorgrids,
ymin=-0.000904109003, ymax=0.019691942143,
ytick style={color=black}
]
\addplot [line width=2pt, navy, mark=*, mark size=3, mark options={solid}]
table {%
1 0.0018755758
2 0.0037511516
3 0.0056267274
4 0.0075023032
5 0.009377879
6 0.0112534548
7 0.0131290306
8 0.0150046064
9 0.0168801822
10 0.018755758
};
\addlegendentry{takeBribe()}
\addplot [line width=2pt, green, mark=square*, mark size=3, mark options={solid}]
table {%
1 3.207514e-05
2 6.415028e-05
3 9.622542e-05
4 0.00012830056
5 0.0001603757
6 0.00019245084
7 0.00022452598
8 0.00025660112
9 0.00028867626
10 0.0003207514
};
\addlegendentry{offerBribe()}
\end{axis}

\end{tikzpicture}\label{fig:pay_to_exit_attack_duration}}
    }}%
    \caption{$\mathsf{PayToBias}$ gas measurements}%
    \label{fig:pay_to_bias_gas_measurements}
\end{figure}

\Cref{alg:paytobias} presents the pseudocode for the $\mathsf{PayToBias}$ contract, which establishes a trustless auction market through three core functions. The $\mathsf{offerBribe}()$ function allows a manipulator to initiate an auction for their slot; $\mathsf{bid}()$ enables other parties to participate; and $\mathsf{takeBribe()}$ lets the manipulator claim the reward after verifying the required conditions. A key feature of this design is the on-chain auction itself, which efficiently prices the right to bias the RANDAO output. This auction model is a generalizable pattern that could be applied to our other attacks, enabling the creation of efficient bribery markets without needing to deploy multiple contracts.

\begin{algorithm}[H]
\caption{The $\mathsf{PayToBias}$ contract: The pseudocode of the $\texttt{offerBribe}()$ and $\texttt{takeBribe}()$ functions as called by the briber and bribee, respectively.}
\label{alg:paytobias}
\begin{algorithmic}[1]
\STATE \textbf{function} offerBribe($\pk \in \mathbb{G}_1,e \in \mathbb{N},\sigma \in \mathbb{G}_2$)
  \STATE\hspace{\algorithmicindent} \textbf{assert} $\mathsf{BLS.Verify}(\pk,e,\sigma)$ // The beacon manipulator releases prematurely its randomness contribution $\sigma$. More precisely, $H(\sigma)$ is the randomness contribution.
  \STATE\hspace{\algorithmicindent}

\STATE \textbf{function} bid($withhold \in \mathbb{B}, bidAmnt \in \mathbb{N}, e \in \mathbb{N}$)
  \STATE\hspace{\algorithmicindent} $auction \gets auctions[e]$ // Get bribery auction corresponding to epoch $e$
  \STATE\hspace{\algorithmicindent} $auction[withhold].bid += bidAmnt$
  \STATE\hspace{\algorithmicindent}

\STATE \textbf{function} takeBribe($e \in \mathbb{N}$, \ $\textit{B}_1, \textit{B}_2$ $\in$ \texttt{BlockHeader})
\STATE $auction \gets auction[e]$
\IF{$block.timestamp < auction.deadline$}
    \STATE \textbf{assert} H(\textit{B}$_1$) = \textbf{blockhash}($n$) and H(\textit{B}$_2$) = \textbf{blockhash}($n$+1)
    \STATE \textbf{assert} $\textit{B}_2\text{.}parent$ = H(\textit{B}$_1$)
    \IF{$\textit{B}_2.timestamp-\textit{B}_1.timestamp > 12$}
        \STATE Transfer $\mathsf{bribeAmnt}$ ether to the winner of the auction
    \ENDIF
\ELSE
    \STATE Transfer $\mathsf{bribeAmnt}$ ether to the winner of the auction
\ENDIF
\end{algorithmic}
\end{algorithm}

\section{$\mathsf{PayToAttest}$ sequence diagrams}\label{sec:pay_to_attest_incentives}

We provide two sequence diagrams illustrating the mechanics of \textit{ex-ante} (\cf~\Cref{fig:pay_to_attest_ex_ante_timeline}) and \textit{ex-post} (\cf~\Cref{fig:pay_to_attest_ex_post_timeline}) reorganizations that use a $\mathsf{PayToAttest}$ contract to buy consensus votes. For clarity, both diagrams depict the simplest forking scenario to which our attacks apply.

\begin{figure}
\begin{tikzpicture}[
  lifeline/.style={thick},
  event/.style={
    font=\small,
    draw,
    dashed,
    fill=white,
    inner sep=2pt,
    rectangle
  },
  msg/.style={-{Stealth[length=2mm]},thick},
  msgshift/.style={midway,above=4pt,font=\scriptsize}, % label shifted up
  msgnormal/.style={midway,above,font=\scriptsize},     % normal label
  phase/.style={
    font=\scriptsize,
    anchor=east,
    rotate=90
  },
  player/.style={
    font=\large % bigger player names
  }
]
% Parameters
\def\nslots{3}
\def\slotheight{3.0} % vertical spacing between horizontal lines
\def\playersep{1.6cm}

% --- Player positions ---
\node[player] (p1) {\textcolor{bribeegreen}{\textbf{B}ribee}};
\node[player,right=\playersep of p1] (p2) {\textcolor{adversaryred}{\textbf{A}dversary}};
\node[player,right=\playersep of p2] (p3) {\textcolor{honestblue}{\textbf{H}onest}};
\node[player,right=\playersep-1.5cm of p3,align=center] (p4) {$\mathsf{PayToAttest}$\\contract};

% --- Lifelines ---
\draw[lifeline] (p1) -- ++(0,-\nslots*\slotheight);
\draw[lifeline] (p2) -- ++(0,-\nslots*\slotheight);
\draw[lifeline] (p3) -- ++(0,-\nslots*\slotheight);
\draw[lifeline] (p4) -- ++(0,-\nslots*\slotheight);

% --- Slot labels (no top dashed line) ---
\foreach \i/\label/\tag in {
  0/{Slot $n+1$}/{$\advblock{}{}$},
  1/{Slot $n+2$}/{$\honestblock{}{}$},
  2/{Slot $n+3$}/{$\advblock{}{}$}
} {
  \pgfmathsetmacro{\y}{-(\i+0.25)*\slotheight}
  % main slot label
  \node[phase] at ([xshift=-0.95cm]p1.west |- 0,\y) {\label};
  % extra "POC" tag, shifted left and also DOWN by 0.5cm
  \node[phase] at ([xshift=-1.2cm,yshift=-0.4cm]p1.west |- 0,\y) {\tag};
}

% --- Dashed horizontal lines (start after first slot label) ---
\foreach \i in {1,...,\nslots} {
  \pgfmathsetmacro{\y}{-\i*\slotheight}
  \draw[dashed] ([xshift=-0.1cm]p1.west |- 0,\y) -- ([xshift=-0.4cm]p4.east |- 0,\y);
}

% --- Slot content ---

% Slot 1
\pgfmathsetmacro{\ycenter}{-\slotheight/2}
\pgfmathsetmacro{\ytop}{-\slotheight*1/2}
\pgfmathsetmacro{\ybottom}{-\slotheight*1}
\pgfmathsetmacro{\gap}{(\ybottom-\ytop)/3}

\node[event] at (p2 |- 0,\ytop-\gap*1.5) {Build block};
\draw[msg] (p2 |- 0,\ytop-0.4*\gap) -- (p1 |- 0,\ytop-0.4*\gap)
  node[pos=0.5,above,font=\scriptsize]{Send $H(m)$};
\draw[msg] (p2 |- 0,\ytop+0.7*\gap) -- (p1 |- 0,\ytop+0.7*\gap)
  node[pos=0.5,above,font=\scriptsize]{Send signed $\mathsf{offerBribe}$ tx};

\node[event] at (p1 |- 0,\ytop+1.8*\gap) {No vote};
\node[event] at (p2 |- 0,\ytop+1.8*\gap) {Vote for Slot $n+1$};
\node[event] at (p3 |- 0,\ytop+1.8*\gap) {Vote for Slot $n$};

% Slot 2
\pgfmathsetmacro{\ytop}{-\slotheight*3/2}
\pgfmathsetmacro{\ybottom}{-\slotheight*2}

\node[event] at (p3 |- 0,\ytop-2.0*\gap) {Propose block};

\node[event] at (p1 |- 0,\ytop-0.5*\gap) {No vote};
\node[event] at (p2 |- 0,\ytop-0.5*\gap) {Vote for Slot $n+1$};
\node[event] at (p3 |- 0,\ytop-0.5*\gap) {Vote for Slot $n+2$};

\draw[msg] (p1 |- 0,\ytop+1.0*\gap) -- (p2 |- 0,\ytop+1.0*\gap)
  node[pos=0.5,above,font=\scriptsize]{Send $\takebribe$};

\draw[msg] (p2 |- 0,\ytop+1.9*\gap) -- (p3 |- 0,\ytop+1.9*\gap)
  node[pos=0.5,above,font=\scriptsize]{Send votes of {\textcolor{bribeegreen}{\textbf{B}ribee}}};

% Slot 3
\pgfmathsetmacro{\ycenter}{-(2*\slotheight + \slotheight/2)}
\pgfmathsetmacro{\ytop}{-\slotheight*5/2}
\pgfmathsetmacro{\ybottom}{-\slotheight*3}

\node[event] at (p2 |- 0,\ytop-1.8*\gap) {Propose block};

\draw[msg] (p2 |- 0,\ytop-0.0*\gap) -- (p4 |- 0,\ytop-0.0*\gap)
  node[pos=0.31,above,font=\scriptsize]{Send $\offerbribe$, $\takebribe$};

\draw[msg] (p4 |- 0,\ytop+1.5*\gap) -- (p1 |- 0,\ytop+1.5*\gap)
  node[pos=0.12,above,font=\scriptsize]{Pays reward};

\end{tikzpicture}
\caption{Timeline of \emph{an ex-ante reorg} attack where the attacker leverages a $\mathsf{PayToAttest}$ bribery contract. Time increases from top to bottom. In Slot $n+1$, after the block is built, the briber broadcasts the $\offerbribe()$ transaction to the bribee. The bribee withholds his votes, but broadcasts the signatures necessary for them in the $\takebribe()$ transaction during Slot $n+2$. The adversary sends the aggregated votes to the honest validators. Finally, the adversary finishes the fork by building on the withheld block. The block contains the timely $\takebribe()$ transaction, so the bribee receives the reward.}
\label{fig:pay_to_attest_ex_ante_timeline}
\end{figure}
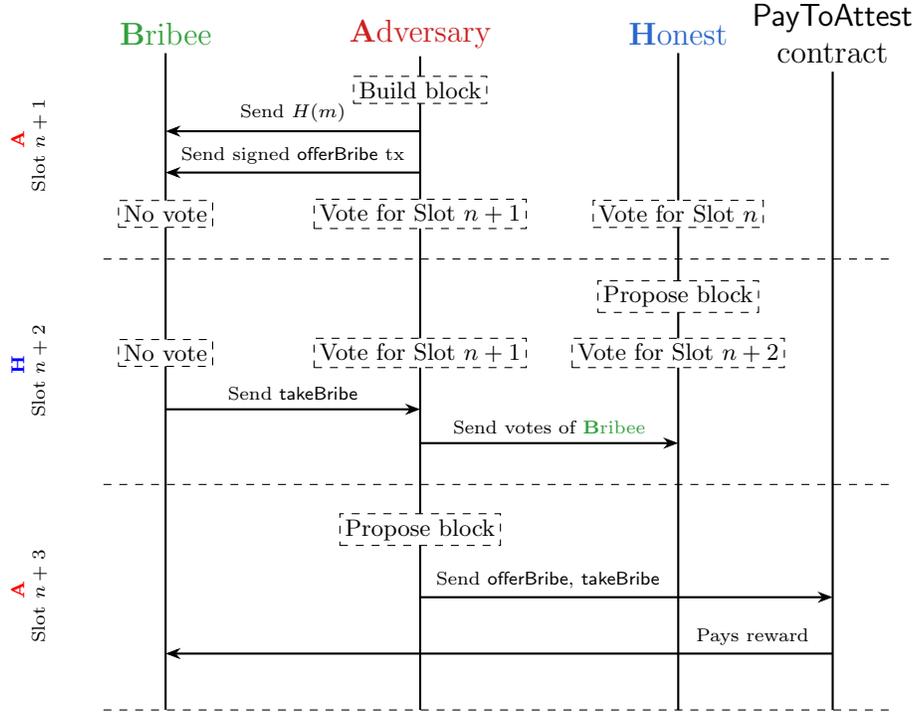

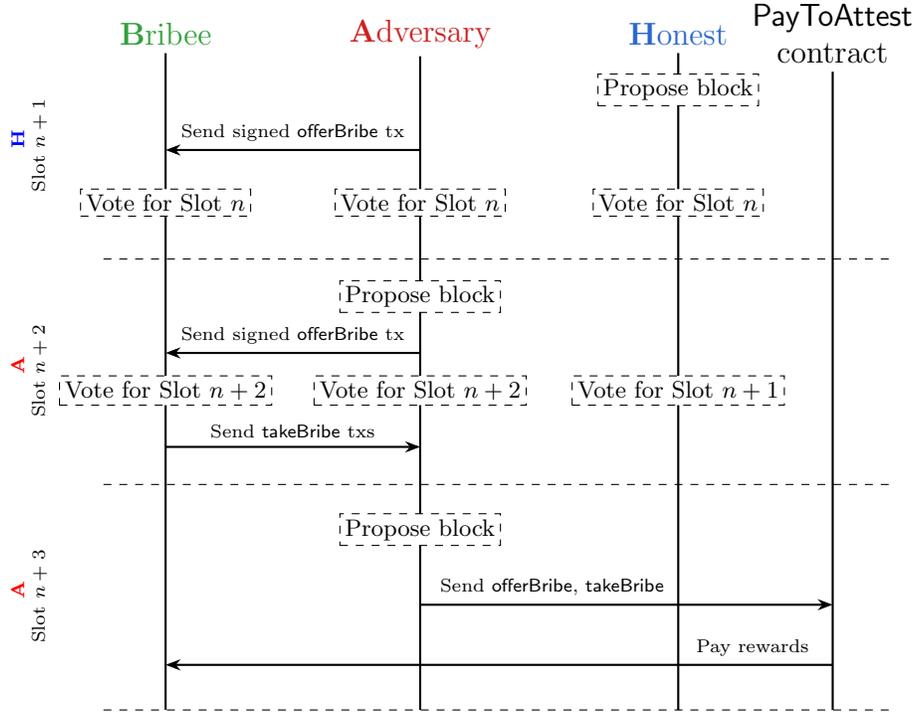
\begin{figure}
\begin{tikzpicture}[
  lifeline/.style={thick},
  event/.style={
    font=\small,
    draw,
    dashed,
    fill=white,
    inner sep=2pt,
    rectangle
  },
  msg/.style={-{Stealth[length=2mm]},thick},
  msgshift/.style={midway,above=4pt,font=\scriptsize}, % label shifted up
  msgnormal/.style={midway,above,font=\scriptsize},     % normal label
  phase/.style={
    font=\scriptsize,
    anchor=east,
    rotate=90
  },
  player/.style={
    font=\large % bigger player names
  }
]
% Parameters
\def\nslots{3}
\def\slotheight{3} % vertical spacing between horizontal lines
\def\playersep{1.6cm}

% --- Player positions ---
\node[player] (p1) {\textcolor{bribeegreen}{\textbf{B}ribee}};
\node[player,right=\playersep of p1] (p2) {\textcolor{adversaryred}{\textbf{A}dversary}};
\node[player,right=\playersep of p2] (p3) {\textcolor{honestblue}{\textbf{H}onest}};
\node[player,right=\playersep-1.5cm of p3,align=center] (p4) {$\mathsf{PayToAttest}$\\contract};

% --- Lifelines ---
\draw[lifeline] (p1) -- ++(0,-\nslots*\slotheight);
\draw[lifeline] (p2) -- ++(0,-\nslots*\slotheight);
\draw[lifeline] (p3) -- ++(0,-\nslots*\slotheight);
\draw[lifeline] (p4) -- ++(0,-\nslots*\slotheight);

% --- Slot labels (no top dashed line) ---
\foreach \i/\label/\tag in {
  0/{Slot $n+1$}/{$\honestblock{}{}$},
  1/{Slot $n+2$}/{$\advblock{}{}$},
  2/{Slot $n+3$}/{$\advblock{}{}$}
} {
  \pgfmathsetmacro{\y}{-(\i+0.25)*\slotheight}
  % main slot label
  \node[phase] at ([xshift=-0.95cm]p1.west |- 0,\y) {\label};
  % extra "POC" tag, shifted left and also DOWN by 0.5cm
  \node[phase] at ([xshift=-1.2cm,yshift=-0.4cm]p1.west |- 0,\y) {\tag};
}
% --- Dashed horizontal lines (start after first slot label) ---
\foreach \i in {1,...,\nslots} {
  \pgfmathsetmacro{\y}{-\i*\slotheight}
  \draw[dashed] ([xshift=-0.1cm]p1.west |- 0,\y) -- ([xshift=-0.4cm]p4.east |- 0,\y);
}

% --- Slot content ---

% Slot 1
\pgfmathsetmacro{\ycenter}{-\slotheight/2}
\pgfmathsetmacro{\ytop}{-\slotheight*1/2}
\pgfmathsetmacro{\ybottom}{-\slotheight*1}
\pgfmathsetmacro{\gap}{(\ybottom-\ytop)/3}

\node[event] at (p3 |- 0,\ytop-\gap*1.5) {Propose block};

\draw[msg] (p2 |- 0,\ytop+0.1*\gap) -- (p1 |- 0,\ytop+0.1*\gap)
  node[pos=0.5,above,font=\scriptsize]{Send signed $\offerbribe$ tx};

\node[event] at (p1 |- 0,\ytop+1.5*\gap) {Vote for Slot $n$};
\node[event] at (p2 |- 0,\ytop+1.5*\gap) {Vote for Slot $n$};
\node[event] at (p3 |- 0,\ytop+1.5*\gap) {Vote for Slot $n$};

% Slot 2
\pgfmathsetmacro{\ytop}{-\slotheight*3/2}
\pgfmathsetmacro{\ybottom}{-\slotheight*2}

\node[event] at (p2 |- 0,\ytop-2.0*\gap) {Propose block};

\draw[msg] (p2 |- 0,\ytop-0.5*\gap) -- (p1 |- 0,\ytop-0.5*\gap)
  node[pos=0.5,above,font=\scriptsize]{Send signed $\offerbribe$ tx};

\node[event] at (p1 |- 0,\ytop+0.5*\gap) {Vote for Slot $n+2$};
\node[event] at (p2 |- 0,\ytop+0.5*\gap) {Vote for Slot $n+2$};
\node[event] at (p3 |- 0,\ytop+0.5*\gap) {Vote for Slot $n+1$};

\draw[msg] (p1 |- 0,\ytop+2.0*\gap) -- (p2 |- 0,\ytop+2.0*\gap)
  node[pos=0.5,above,font=\scriptsize]{Send $\takebribe$ txs};

% Slot 3
\pgfmathsetmacro{\ycenter}{-(2*\slotheight + \slotheight/2)}
\pgfmathsetmacro{\ytop}{-\slotheight*5/2}
\pgfmathsetmacro{\ybottom}{-\slotheight*3}

\node[event] at (p2 |- 0,\ytop-1.8*\gap) {Propose block};

\draw[msg] (p2 |- 0,\ytop+0.2*\gap) -- (p4 |- 0,\ytop+0.2*\gap)
  node[pos=0.32,above,font=\scriptsize]{Send $\offerbribe$, $\takebribe$};

\draw[msg] (p4 |- 0,\ytop+1.8*\gap) -- (p1 |- 0,\ytop+1.8*\gap)
  node[pos=0.12,above,font=\scriptsize]{Pay rewards};

\end{tikzpicture}
\caption{Timeline of \emph{an ex-post reorg} attack where the attacker leverages a $\mathsf{PayToAttest}$ bribery contract. Time increases from top to bottom. In Slot $n+1$, the adversary decided to fork out the new $\honestblock{}{}$, possibly due to its high MEV content. Thus, he offers a bribe to vote for the block in Slot $n$ instead, which the bribee accepts. In the next slot, another $\offerbribe()$ is broadcast to vote for the new $\advblock{}{}$. The bribee collaborates again, making the adversarial branch the heaviest, according to the LMD-Ghost rule. The attacker included the $2$ $\offerbribe()$ and the 2 timely $\takebribe()$ transactions, thus, the reward payments are issued. }
\label{fig:pay_to_attest_ex_post_timeline}
\end{figure}

\begin{landscape}
\begin{table}[h!]
\caption{Qualitative comparison of existing bribery attacks on Ethereum PoS.
A property is marked with \cmark \ if it is achieved and with \xmark \ otherwise, -- is used if a property does not apply. If the symbol is within brackets, \eg (\cmark), this
means that this property is achieved (or can be augmented), but this was initially not discussed or considered by the authors.
$\sim$ denotes that the property cannot be clearly mapped to any of the previously defined categories without further details or discussion.}
\centering
\renewcommand{\arraystretch}{1.2}
\setlength{\tabcolsep}{4pt}
\small
\resizebox{\linewidth}{!}{%
\begin{tabular}{lccccccccccc}
\toprule
\textbf{Attack} &
\multicolumn{1}{p{2.5cm}}{\centering Transaction \\ reversal} &
\multicolumn{1}{p{2.5cm}}{\centering Transaction \\ ordering} &
\multicolumn{1}{p{2.5cm}}{\centering Transaction \\ exclusion} &
\multicolumn{1}{p{2.5cm}}{\centering Transaction \\ triggering} & \multicolumn{1}{p{2.5cm}}{\centering Required interference \\ with consensus} &
\multicolumn{1}{p{2cm}}{\centering Requires smart \\ contract} &
\textbf{Payment} &
\multicolumn{1}{p{2cm}}{\centering Trustless for \\ attacker} &
\multicolumn{1}{p{2.3cm}}{\centering Trustless for \\ collaborator} &
\textbf{Subsidy} &
\multicolumn{1}{p{2.5cm}}{\centering Compensates if \\ attack fails} \\ 
\midrule
$\mathsf{PayToAttest}$ (this work) & \cmark & \xmark & (\cmark) & \xmark & Deep fork & \cmark & in-band & \cmark & \cmark & \cmark & (\cmark) \\ 
$\mathsf{PayToExit}$ (this work) & \xmark & \xmark & \xmark & \xmark & No fork & \cmark & in-band & \cmark & \cmark & \cmark & \cmark \\
$\mathsf{PayToBias}$ (this work) & \xmark & \xmark & \xmark & \xmark & No fork & \cmark & in-band & \cmark & \cmark & \cmark & \cmark \\
Simple Attack~\cite{DBLP:conf/eurosp/SarencheTMSP25} & \cmark  & \cmark  & \cmark & \xmark & Fork & \xmark & in-band  & \cmark  & (\cmark) & \cmark & \xmark \\
Strong Simple Attack~\cite{DBLP:conf/eurosp/SarencheTMSP25}& \cmark  & \cmark  & \cmark & \xmark  & Fork & \xmark  & in-band & \cmark & (\cmark) & \cmark & \xmark \\
Extended Attack~\cite{DBLP:conf/eurosp/SarencheTMSP25}& \cmark  & \cmark  & \cmark & \xmark  & Deep fork  & \xmark & in-band & \cmark & (\cmark) & \cmark & \xmark \\
BriDe Arbitrager~\cite{yang2024bride} & --  & \cmark  & \cmark & \cmark & No fork & \cmark & in-band  & \cmark  & (\cmark) & \cmark & \cmark \\
Guided Bribing~\cite{karakostas2024blockchain} & \cmark & \cmark & \cmark & \xmark & No fork & \xmark & out-of-band & \xmark & \xmark & \cmark & \cmark \\
Effective Bribing~\cite{karakostas2024blockchain} & \cmark & \cmark & \cmark & \xmark & No fork & \xmark & out-of-band & \xmark & (\cmark) & \cmark & \xmark \\
Bribe \& Fork~\cite{DBLP:conf/aft/AvarikiotiKLM24} & \cmark & \cmark & \cmark & \cmark & Deep fork & \cmark & in-band & \cmark & \cmark & \cmark & \cmark \\
Bribery Semi-Selfish Mining~\cite{hu2023novel} & \xmark & \cmark & \cmark & \xmark & Deep fork & \xmark & out-of-band & \xmark & \xmark & \xmark & \xmark \\
Bribery Stubborn Mining~\cite{hu2023novel} & \xmark & \cmark & \cmark & \xmark & Fork & \cmark & out-of-band & \xmark & \xmark & \xmark & \xmark \\
CensorshipCon~\cite{mccorry2018smart} & \xmark & (\cmark) & \cmark & (\cmark) & No forks & \cmark & in-band & $\sim$ & \xmark & \cmark & \xmark \\ 
HistoryRevisionCon~\cite{mccorry2018smart} & \cmark & \xmark & \cmark & (\cmark) & Deep fork & \cmark & in-band & \cmark & $\sim$ & \cmark & \xmark \\ 
GoldfingerCon~\cite{mccorry2018smart} & -- & -- & \cmark all & (\cmark) & No fork & \cmark & out-of-band & \cmark & \xmark & \cmark & \xmark \\ 
P2W Tx Excl. \& Ord~\cite{judmayer2019pay} & \xmark & \cmark & \cmark & (\cmark) & No fork & \cmark & out-of-band & \cmark & \cmark & \xmark & \cmark \\ 
P2W Tx Rev., Excl.  \& Ord.~\cite{judmayer2019pay} & \cmark & \cmark & \cmark & (\cmark) & Deep fork & \cmark & out-of-band & \cmark & \cmark & \xmark & \cmark \\ 
P2W Tx Ord.~\cite{judmayer2019pay}  & \xmark & \cmark & \xmark & (\cmark) & No fork & \cmark & in-band & \cmark & \cmark & \xmark & \xmark \\
P2W Tx Excl.~\cite{judmayer2019pay}  & \xmark & \xmark & \cmark & (\cmark) & No fork & \cmark & in-band & \cmark & \cmark & \xmark & \xmark \\  %https://eprint.iacr.org/2019/775.pdf
\bottomrule
\end{tabular}%
}

\label{tab:attacks_qualitative_comparison}
\end{table}
\end{landscape}
\section{$\mathsf{PayToAttest}$ for General Forking Attacks}\label{sec:exante_paytofork}

Note the original in-band payment idea using anyone-can-spend addresses for forking as proposed in~\cite{bonneau2016buy} is not applicable in the Ethereum case. Ethereum blocks gain weight from the attestors, and attestators are typically not the same entities as the block proposers. Thus, anyone-can-spend addresses or simple transfers do not work.

%The signed block header $H(m)$ could be fed into the bribery contract by the bribee calling the $\mathsf{takeBribe()}$ function \emph{after} a fork happens. Then, the contract must also verify the correctness of the signed message $m$. We refer the details to~\Cref{sec:paytoattest_implementation}.

$\mathsf{PayToAttest}$ bribery contracts can be used to perform reorg attacks for a much larger family of chain strings than was shown in~\Cref{sec:pay_to_fork}.

\begin{figure}
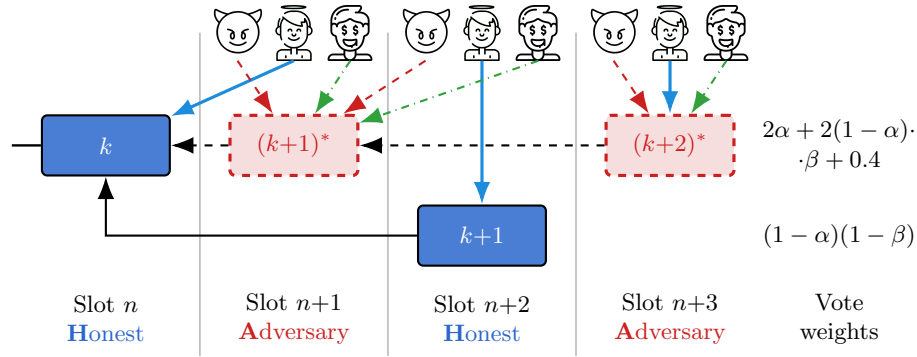

\begin{center}
\begin{tikzpicture}[font=\small]

% Horizontal y level for main chain
\def\y{0}
\def\x{0}

% Slot dividing vertical lines (four slots)
\foreach \i in {0,1,2} {
  \draw[slotline] (2.5*\i,\y-1.8) -- (2.5*\i,\y+2.5);
}

% Attestors
\node[inner sep=0pt] (b1) at (.5,2.5)
    {\includegraphics[width=0.06\textwidth]{Figures/b.png}};
\node[inner sep=0pt] (a1) at (1.25,2.5)
    {\includegraphics[width=0.06\textwidth]{Figures/a.png}};
\node[inner sep=0pt] (r1) at (2.0,2.5)
    {\includegraphics[width=0.06\textwidth]{Figures/r.png}};
    
\node[inner sep=0pt] (b2) at (3,2.5)
    {\includegraphics[width=0.06\textwidth]{Figures/b.png}};
\node[inner sep=0pt] (a2) at (3.75,2.5)
    {\includegraphics[width=0.06\textwidth]{Figures/a.png}};
\node[inner sep=0pt] (r2) at (4.5,2.5)
    {\includegraphics[width=0.06\textwidth]{Figures/r.png}};
    
\node[inner sep=0pt] (b3) at (5.5,2.5)
    {\includegraphics[width=0.06\textwidth]{Figures/b.png}};
\node[inner sep=0pt] (a3) at (6.25,2.5)
    {\includegraphics[width=0.06\textwidth]{Figures/a.png}};
\node[inner sep=0pt] (r3) at (7,2.5)
    {\includegraphics[width=0.06\textwidth]{Figures/r.png}};

%Blocks
\node[honest]    (k)      at (-1.25,1)            {$k$};
\node[adversary] (k1a)    at (1.25,1)       {$(k{+}1)^*$};
\node[honest]    (k1)     at (3.75,-0.2)            {$k{+}1$};
\node[adversary] (k2a)    at (6.25,1)       {$(k{+}2)^*$};

% Attestation arrows
\draw[honestarw] (a1.south) to (k);
\draw[advarw] (b1.south) to (k1a);
\draw[altarw] (r1.south) to (k1a);

\draw[honestarw] (a2.south) to (k1);
\draw[advarw] (b2.south) to (k1a);
\draw[altarw] (r2.south) to (k1a);

\draw[honestarw] (a3.south) to (k2a);
\draw[advarw] (b3.south) to (k2a);
\draw[altarw] (r3.south) to (k2a);

%Chain arrows
\draw[black, line width=0.9pt] (-2.5,1) -- (k);
\draw[achain] (k1a) -- (k);
\draw[achain] (k2a) -- (k1a);
\draw[chain] (k1) -- (-1.25,-.2) -- (k);

% Slot labels under diagram
\node[align=center] at (-1.25,\y-1.3)  {Slot $n$ \\ \textcolor{honestblue}{\textbf{H}onest}};
\node[align=center] at (1.25,\y-1.3)  {Slot $n{+}1$ \\ \textcolor{adversaryred}{\textbf{A}dversary}};
\node[align=center] at (3.75,\y-1.3)  {Slot $n{+}2$ \\ \textcolor{honestblue}{\textbf{H}onest}};
\node[align=center] at (6.25,\y-1.3)  {Slot $n{+}3$ \\ \textcolor{adversaryred}{\textbf{A}dversary}};

\node[align=center] at (8.5,1) {$2\alpha+2(1-\alpha)\cdot$\\$\cdot\beta+0.4$};
\node[align=center] at (8.5,\y-0.2) {$(1-\alpha)(1-\beta)$};
\node[align=center] at (8.5,\y-1.3) {Vote\\ weights}; 
\end{tikzpicture}
\end{center}
\caption{Performing an ex-ante reorg with a $\mathsf{PayToAttest}$ bribery contract. Colored arrows indicate which blocks different validators vote for as the head of the blockchain. Black arrows represent hash pointers. Red (blue) blocks are proposed by the briber (honest validators). }
\label{fig:ex-ante-fork}
\end{figure}

\begin{figure}[ht]
    \centering
   \resizebox{0.99\textwidth}{!}{\input{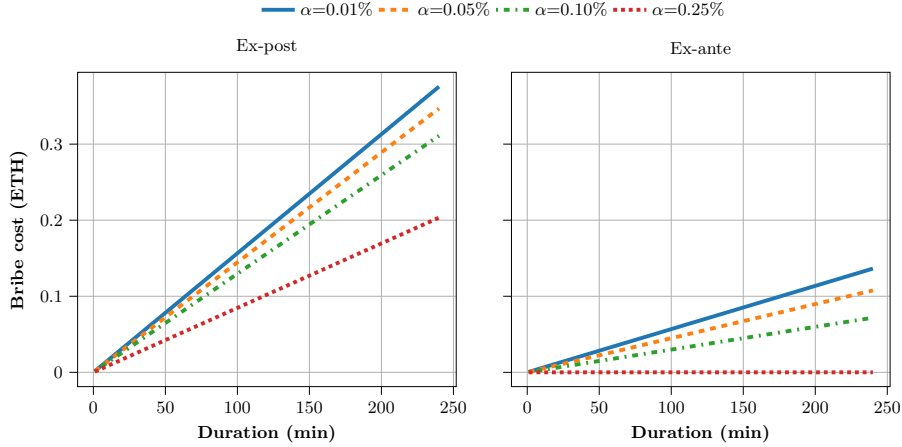}}
    \caption{Ex-post and ex-ante reorg bribe costs for continuous short-term reorg attacks. In these measurements, we assume that the briber leverages the $\mathsf{PayToAttest}$ bribery contract to bribe and reorg honest blocks at every possible occasion. Moreover, we assume that the minimal amount of rational validators accept the bribe necessary for the reorg to be successful,~\cf~\Cref{eq:expost_winning_condition}.}
    \label{fig:reorg_bribe_costs}
\end{figure}

\subsection{General ex-post reorgs}\label{sec:general_expost_reorgs}

In~\Cref{sec:pay_to_fork}, we saw that the adversary can execute an ex-post reorg leveraging a $\mathsf{PayToAttest}$ contract to buy rational validators' attestations. Moreover, a motivated adversary can execute ex-post reorgs for many more chain strings as we show next.

\begin{theorem}
    Given the following chain string $\honestblock{}{h}\advblock{}{a}$, the adversary can fork out $\honestblock{}{h}$ with the joint voting power of the adversary and bribees if~\cref{eqn:ex-post-suff-votes} holds.
    More formally, the adversary's fork wins iff.
    \begin{equation}\label{eqn:ex-post-suff-votes}
        (a - 1)\cdot(\alpha + (1 - \alpha)\beta) + \pboost{} > (h + a - 1)\cdot(1 - (\alpha + (1 - \alpha)\beta))\enspace .
\end{equation}
\end{theorem}

\begin{proof}
    First, the honest branch $\honestblock{}{h}$ is proposed, receiving a total attestation of $h\cdot(1 - (\alpha + (1 - \alpha)\beta))$. During this period, the adversary and the bribees vote for block $k$ (\cf~\cref{fig:ex_post_fork_explainer}). In the next $a-1$ slots they start the adversarial branch on top of block $k$, continuously voting for it: $(a - 1)\cdot(\alpha + (1 - \alpha)\beta)$. ``Worst''-case scenario, the honest attestators still consider $\honestblock{}{h}$ as the canonical chain, thus getting votes of $(a - 1)\cdot(1 - (\alpha + (1 - \alpha)\beta))$. Finally, the adversary builds and proposes the last block of the adversarial branch promptly, giving it the proposer boost $\pboost{}$. To calculate when the adversarial branch is considered the canonical chain by the LMD-Ghost rule, we get~\Cref{eqn:ex-post-suff-votes}.
\end{proof}

\subsection{General ex-ante reorgs}\label{sec:general_exante_reorgs}

Ex-ante reorgs against the Ethereum PoS protocol were introduced in~\cite{schwarz2022three}. Utilizing our $\mathsf{PayToAttest}$ bribery contracts, an adversary can launch ex-ante reorgs for various chain strings as we show below.

\begin{theorem}[Condition for ex-ante reorgs]\label{thm:ex-ante-reorgs}
    Given the following chain string $\advblock{}{a}\honestblock{}{h}\advblock{}{}$, the adversary is capable of forking out $\honestblock{}{h}$ with the joint voting power of the adversary and the collaborating bribees if the following holds.
    \begin{equation}\label{eqn:ex-ante-suff-votes}
        (a+h)\cdot(\alpha + (1 - \alpha)\beta) + \pboost{} > h\cdot(1-(\alpha + (1 - \alpha)\beta))\enspace .
    \end{equation}
\end{theorem}

\begin{proof}
    During the first $a$ slots, the adversary secretly builds $\advblock{}{a}$ blocks, which receive $(a + h)\cdot(\alpha + (1 - \alpha)\beta)$ votes by the end of the honest branch. On the other hand, the honest validators consider the private blocks $\advblock{}{a}$ missed, voting for block $k$ instead,~\cf~\Cref{fig:ex-ante-fork}.
    The first honest block is built on top of block $k$, and the honest branch gets a total of $h\cdot(1-(\alpha + (1 - \alpha)\beta))$ attestations. Finally, the secret branch is published when the adversary builds on top of $\advblock{}{a}$, receiving the proposer boost. The branch is heavier whenever~\Cref{eqn:ex-ante-suff-votes} holds.
\end{proof}
\section{Additional proofs and calculations}\label{sec:additional_simple_calculations}

\subsection{Present value (PV) half-life for a fixed discount factor}
\label{sec:halflife}

Let \(R\) denote the expected nominal annual staking reward.  Consider a staker who would otherwise receive \(R\) each year in perpetuity.  The present value ($PV$) of the infinite stream (a perpetuity) discounted at rate \(r>0\) is
\begin{equation}
    PV_{\infty} \;=\; \frac{R}{r}\enspace,
\end{equation}
as standard in financial mathematics~\cite{Brealey2020}.

If the staker instead receives rewards only up to a finite horizon \(T\), the PV of the finite annuity is
\begin{equation}
    PV(T) \;=\; R\sum_{t=1}^{T} \frac{1}{(1+r)^t}
           \;=\; R\frac{1-(1+r)^{-T}}{r}\enspace.
\end{equation}

The \emph{$PV$ half-life} is defined as the horizon \(T_{1/2}\) at which the $PV$ of cumulative rewards equals one half of the perpetuity value. This yields
\begin{equation}
    T_{1/2} \;=\; \frac{\ln 2}{\ln(1+r)}.
\end{equation}

For \(r=0.08\), we obtain \(T_{1/2} \approx 9\) years, implying that half of the $PV$ of all future staking rewards is realized within the first decade of operation.  This provides a transparent benchmark for comparing one-time bribes with the net present value of forgone staking income.

\subsection{Visualizing the inefficacy of the constant-bribe upper bound}\label{sec:apr-rect}
Recall that in~\Cref{sec:pay_to_exit_incentives}, we computed an optimal \emph{constant} bribe $b^{*}$, \ie in total, a briber pays $k^*\cdot b^{*}$ bribe amount to $k^*$ exiting validators. However, Ethereum protocol rewards are set by the consensus protocol as a decreasing function in the number of validators,~\cf~\Cref{eq:util_L} and~\Cref{fig:apr-rect}. This means that the briber must pay increasing bribe amounts to the exiting rational validators since the annual return increases as the number of remaining validators decreases. As~\Cref{fig:apr-rect} illustrates, our calculated optimal bribe is the bribe paid to the last exiting validator from the total exiting $k^*$ validators. Thus, in a more realistic setting where the $\mathsf{PayToBias}$ contract pays dynamic bribe amounts, the adversary only needs to pay the area under the curve as opposed to the rectangle with $b^{*}$ corresponding to a constant bribe amount $b_{max}$,~\cf~\Cref{fig:apr-rect}. We leave it to future work to add the dynamic bribe amount feature to our $\mathsf{PayToExit}$ bribery contract implementation, depending on the number of remaining validators.  
\begin{figure}[H]
  \centering
  \input{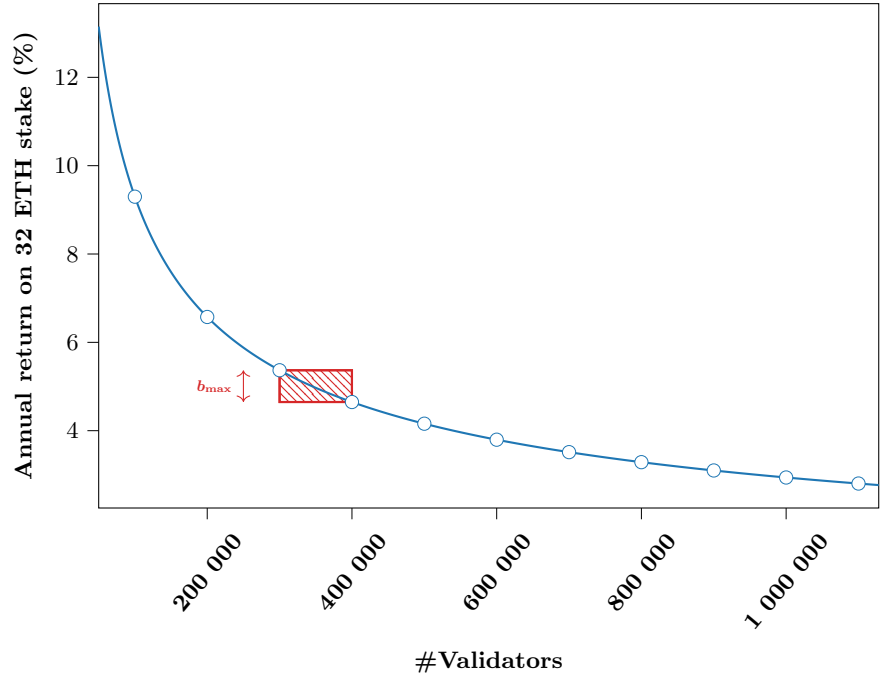}
  \caption{\textbf{Upper-bound rectangle on APR vs.\ $\#$Validators.}
  On the interval $[N_0-k,N_0]$ we mark a constant-bribe upper bound $b_{\max}=R(N-k) \cdot \mathrm{PV}(r,Y)$.
  The rectangle’s area $k\cdot b_{\max}$ is a conservative bribe cost upper bound, \emph{not} an integral under the curve which could be achieved by dynamically changing $\mathsf{PayToExit}$ bribe amounts depending on the number of remaining active validators (left for future work).}
  \label{fig:apr-rect}
\end{figure}

%\subsection{Numerical example -- visual aid}
%\label{sec:num_example_visual_aid}

%\begin{figure}[H]
%  \centering
%  \input{Figures/Numerical_example_5.2}
%  \caption{Enlarged view of the $g(k)$\kamilla{should be $R(N-k) \cdot \mathrm{PV}(r,Y)$} function near the optimum $k^{*}$. The dashed line marks $k^{*}$ exiting validators in the Nash-equilibrium (NE) for the parameters $(N=\num{1127967}, r=8\%,Y=9, \alpha=23\%, \alpha^*=33\%)$; $\color{red}{\times}$ denotes the optimum bribe $b^{*}$ and $\color{blue}{\times}$ denotes $g(k^{*})$\kamilla{should be $R(N-k) \cdot \mathrm{PV}(r,Y)$}.}

%  \label{fig:gkstar}
%\end{figure}

\end{document}